\documentclass[11pt,fleqn]{article}
\usepackage{amssymb,latexsym,amsmath,amsfonts,graphicx, ebezier}

    \advance\textheight by \topskip

    \numberwithin{equation}{section}

\makeatletter
\def\eqalign#1{\null\vcenter{\def\\{\cr}\openup\jot\m@th
  \ialign{\strut$\displaystyle{##}$\hfil&$\displaystyle{{}##}$\hfil
      \crcr#1\crcr}}\,}
\makeatother

\newcommand{\be}{\begin{equation}}
\newcommand{\ee}{\end{equation}}

\newcommand{\wt}{\widetilde}
\newcommand{\de}{\delta}

\newcommand{\al}{\alpha}
\newcommand{\bt}{\beta}
\newcommand{\ga}{\gamma}
\newcommand{\Ga}{\Gamma}

\newcommand{\si}{\sigma}
\newcommand{\Si}{\Sigma}

\newcommand{\Om}{\Omega}
\newcommand{\lb}{\lambda}

\newcommand{\ep}{\varepsilon }

    \def\e{{\epsilon}}
    
    \def\tr{{\rm tr \,}}
    \def\Re{{\rm Re \,}}
    \def\Im{{\rm Im \,}}

    \def\bigO{{\cal O}}

    \def\P2n{{\rm P}_{{\rm II}}^{(n)}}

    \newtheorem{theorem}{Theorem}[section]
    \newtheorem{lemma}[theorem]{Lemma}
    
    \newtheorem{proposition}[theorem]{Proposition}
    
    \newtheorem{Definition}[theorem]{Definition}
    
    \newtheorem{Remark}[theorem]{Remark}
    \newenvironment{remark}{\begin{Remark}\rm}{\end{Remark}}
    \newtheorem{Example}[theorem]{Example}
    
    \newtheorem{Assumptions}[theorem]{Assumptions}

    \newenvironment{proof}%
    {\rm \trivlist \item[\hskip \labelsep{\bf Proof. }]}%
    {\hspace*{\fill}$\Box$\endtrivlist}
    \newenvironment{varproof}%
    {\rm \trivlist \item[\hskip \labelsep{\bf Proof}]}%
    {\hspace*{\fill}$\Box$\endtrivlist}

\begin{document}
\title{Emergence of a singularity for Toeplitz determinants
and Painlev\'e V}
\author{T. Claeys, A. Its, I. Krasovsky}
\maketitle

\begin{abstract}
We obtain asymptotic expansions for Toeplitz determinants corresponding to a family
of symbols depending on a parameter $t$. For $t$ positive, the symbols are regular
so that the determinants obey Szeg\H{o}'s strong limit theorem. If $t=0$, the symbol
possesses a Fisher-Hartwig singularity. Letting $t\to 0$ we analyze the emergence of a
Fisher-Hartwig singularity and a transition between the two different types of
asymptotic behavior for Toeplitz determinants. This transition is described by a
special Painlev\'e V transcendent. A particular case of our result
complements the classical description of Wu, McCoy, Tracy, and Barouch of  the
behavior of a 2-spin correlation function for a large distance between spins
in the two-dimensional Ising model as the phase transition occurs.
\end{abstract}

\section{Introduction}
Consider the Toeplitz determinant with symbol $f(z)\in L^1(C)$, where $C$ is the
unit circle:
\begin{equation}\label{Toeplitz}
D_n=\det(f_{j-k})_{j,k=0}^{n-1}, \qquad
f_{j}=\frac{1}{2\pi}\int_{0}^{2\pi}f(e^{i\theta})e^{-ij\theta}d\theta.
\end{equation}
We are interested in the behavior of $D_n$ as $n\to\infty$.

If $\ln f(z)$ is sufficiently smooth on the unit circle (in particular,
$f(z)$ is never zero for $z\in C$ and has no winding around the origin) so that
$\ln f(z)\in L^1(C)$ and the sum
\be\label{Sf}
\sum_{k=-\infty}^{\infty}|k||(\ln f)_k|^2,\qquad
(\ln f)_k=\frac{1}{2\pi}\int_{0}^{2\pi}\ln f(e^{i\theta})e^{-ik\theta}d\theta,
\ee
converges, then
the asymptotics of $D_n$ are given by the strong Szeg\H o limit theorem \cite{I,GI,J}:
\begin{equation}\label{Sz}
\ln D_n=\frac{n}{2\pi}\int_{0}^{2\pi}\ln f(e^{i\theta})d\theta
+\sum_{k=1}^{\infty}k(\ln f)_k (\ln f)_{-k}+o(1),
\qquad \mbox{ as $n\to\infty$.}
\end{equation}

However, one often encounters a situation where the symbol possesses so-called
Fisher-Hartwig singularities.
In the case of only one such singularity, located
at $z=1$, the symbol has the form:
\begin{equation}\label{symbolFH}
f(z)=|z-1|^{2\al}z^\bt e^{-i\pi\bt}e^{V(z)}=
(2-2\cos\theta)^\alpha e^{i\beta(\theta-\pi)}e^{V(e^{i\theta})}, \qquad \mbox{ for
$0<\theta<2\pi$,}
\end{equation}
where $V(z)$ is a sufficiently smooth function (see \cite{DIK2}) on the unit circle.
The singularity at $z=1$ combines a
jump-type (for $\alpha=0$, $\beta\neq 0$) and a root-type singularity
(for $\beta=0$, $\alpha\neq 0$). For this symbol the sum (\ref{Sf}) diverges,
and therefore Szeg\H o's theorem does not hold. The asymptotics for the Toeplitz
determinant are given instead by the expression \cite{BS, W, Basor, Ehr, DIK2}:
\begin{multline}\label{FH}
\ln D_n=nV_0+\sum_{k=1}^{\infty}kV_kV_{-k}
-(\alpha-\beta)\sum_{k=1}^{\infty}V_k-(\alpha+\beta)\sum_{k=1}^{\infty}V_{-k}\\
+(\alpha^2-\beta^2)\ln n
+\ln\frac{G(1+\alpha+\beta)G(1+\alpha-\beta)}{G(1+2\alpha)}+o(1), \qquad\mbox{ as
$n\to\infty$,}
\end{multline}
if
\be\label{ab}
\al\pm\bt\neq -1,-2,...,
\ee
with
\[
V_k=\frac{1}{2\pi}\int_0^{2\pi}V(e^{i\theta})e^{-ik\theta}d\theta.
\]
Here $G$ is Barnes' $G$-function, which is an entire function having the properties:
$G(z+1)=\Gamma(z)G(z)$, where $\Gamma(z)$ is Euler's $\Gamma$-function, and $G(1)=1$,
$G(-k)=0$ for $k=0,1,2,\dots$. Note that if $V(z)\equiv 0$,
there exists an explicit expression
for $D_n$ with symbol $|z-1|^{2\al}z^\bt e^{-i\pi\bt}$ in terms of $G$-functions \cite{BS2,BW}.

Suppose now that a symbol depends on a parameter $t$ ($f(z)=f(z;t)$) so that
when $t>0$ the symbol is ``regular'', i.e. Szeg\H o's theorem holds for $D_n(t)$,
while at $t=0$ the symbol has the form (\ref{symbolFH}). The purpose of the present
paper
is to study the transition from (\ref{Sz}) to (\ref{FH}) as $t\to 0$.
Namely, consider the following symbol
\begin{equation}\label{symbol}
f(z)=(z-e^{t})^{\alpha+\beta}(z-e^{-t})^{\alpha-\beta}z^{-\alpha+\beta}
e^{-i\pi(\alpha+\beta)}e^{V(z)},\qquad \al\pm\bt\neq -1,-2,...,
\end{equation}
where $t\geq 0$ is sufficiently small and $\alpha, \beta\in\mathbb C$ with
$\Re\alpha>-\frac{1}{2}$.
We further assume that $V(z)$ is analytic in an annulus containing the unit circle
and write it there in terms of its Fourier series
\begin{equation}V(z)=\sum_{k=-\infty}^{+\infty}V_kz^k\label{V}.\end{equation}
We define the powers in (\ref{symbol}) with arguments between $0$ and $2\pi$.
With this choice of branch cuts, $f$, and moreover $\ln f$,  is analytic in
$\mathbb C\setminus \left([0,e^{-t}]\cup[e^{t},+\infty)\right)$ and, in
particular, on the unit circle for $t>0$.
Therefore, for any fixed $t>0$, the asymptotics of  the Toeplitz determinant $D_n(t)$
are given by (\ref{Sz}). Calculating the Fourier coefficients $(\ln f)_k$, we obtain
\begin{multline}\label{Szego}
\ln
D_n(t)=nt(\alpha+\beta)+nV_0\\+\sum_{k=1}^{\infty}k\left[V_k-(\alpha+\beta)\frac{e^{-tk}}{k}\right]\left[V_{-k}-(\alpha-\beta)\frac{e^{-tk}}{k}\right]+o(1),
\qquad \mbox{ as $n\to\infty$, $t>0$.}
\end{multline}
For $t=0$, the symbol reduces to (\ref{symbolFH}) (with analytic $V$). Therefore,
for $t=0$, the asymptotics of $D_n(t)$ are given by (\ref{FH}).
In the present paper we describe
the transition from (\ref{Szego}) to (\ref{FH}) when $t$ decreases to $0$.

The paper also sets the stage for analysis of various other transition asymptotics
for Toeplitz determinants,
such as 2 singularities approaching each other; emergence of an arc of the unit
circle where
the symbol $f=0$ from 2 jump-type singularities at the ends of the arc, etc.

Our analysis explains to some extent the question of connection between Painlev\'e
tau-functions and Toeplitz determinants which was noticed before: see
\cite{BTpainleve} for a discussion
and references. Historically and as the most prominent example, this connection
appeared in the study
of 2-spin correlation functions in the 2-dimensional Ising model. We discuss this in
the following section.

\subsection{Application: the two-dimensional Ising model}

Transitions between Szeg\H{o} weights and Fisher-Hartwig weights arise for example
in the
theory of solvable two-dimensional statistical models and one-dimensional Heisenberg
spin chains. Recall
the two-dimensional Ising model solved by Onsager (see, e.g., \cite{MW,McCoy}).
In this model a $2{\cal M}\times 2{\cal N}$ rectangular lattice is considered with
an associated
spin variable $\si_{jk}$ taking the values $1$ and $-1$
at each vertex $(j,k)$, $-{\cal M}\le j\le {\cal M}-1$,
$-{\cal N}\le k\le {\cal N}-1$.
There are $2^{4{\cal M}{\cal N}}$ possible spin configurations
$\{\si\}$ of the lattice (a configuration corresponds to values of all
$\si_{jk}$ fixed). We associate with each configuration the energy
of the nearest-neighbor coupling (imposing the cyclic boundary conditions on the lattice)
\be
E(\{\si\})=
-\sum_{j=-{\cal M}}^{{\cal M}-1}\sum_{k=-{\cal N}}^{{\cal N}-1}
\left(\ga_1\si_{jk}\si_{j\,k+1}+\ga_2\si_{jk}\si_{j+1\,k}\right),\qquad
\ga_1,\ga_2>0.
\ee
The partition function at a temperature $T>0$ is equal to
\be
Z(T)=
\sum_{\{\si\}}e^{-E(\{\si\})/T},
\ee
where the sum is over all configurations. A remarkable feature of this model
is the presence of a thermodynamic phase transition in the limit of the
infinite lattice at a certain temperature $T_c$ depending on $\ga_1$, $\ga_2$.

Define a 2-spin correlation function by the expression
\be\label{corr}
<\si_{00}\si_{nn}>=
\lim_{{\cal M},{\cal N}\to\infty}{\frac{1}{Z(T)}}
\sum_{\{\si\}}\si_{00}\si_{nn}e^{-E(\{\si\})/T}.
\ee
For large $n$, this function measures the long-range order
in the lattice at a temperature $T$, which determines magnetization. Indeed
one can show that the spontaneous magnetization $M$ is given by the expression
\be\label{M}
M=\sqrt{\lim_{n\to\infty}<\si_{00}\si_{nn}>}.
\ee

It is a remarkable fact that the 2-spin correlation function is a Toeplitz
determinant
\be\label{Ising}
<\si_{00}\si_{nn}>=e^{nt/2} D_n(t),\qquad
f(z;t)=(z-e^t)^{-1/2}(z-e^{-t})^{1/2}z^{-1/2}e^{i\pi/2},
\ee
where
\be
e^t=\sinh{\frac{2\ga_1}{T}}\sinh{\frac{2\ga_2}{T}},
\ee
and the branches of the roots are chosen with the arguments from $0$ to $2\pi$.
The symbol in (\ref{Ising}) has the form (\ref{symbol}) with
\be\label{Iab}
\al=0,\qquad \bt=-{\frac{1}{2}},\qquad V(z)\equiv 0.
\ee
The critical temperature $T_c$ is defined by the condition that
$t=0$.

For $T<T_c$ we have $t>0$. Therefore the strong Szeg\H o limit theorem (\ref{Szego}) holds,
and using the elementary identity
\be\label{logid}
\sum_{k=1}^\infty e^{-2kt}/k=-\ln(1-e^{-2t}),
\ee
we rederive the well-known result
\begin{multline}\label{D1}
e^{nt/2}D_n(t)=(1-e^{-2t})^{1/4}(1+o(1))=\\
\left[1-\left(\sinh{\frac{2\ga_1}{T}}\sinh{\frac{2\ga_2}{T}}\right)^{-2}
\right]^{1/4}(1+o(1)),\qquad\mbox{ as $n\to\infty$},\ T<T_c.
\end{multline}
The  correlations tend to a constant as $n\to\infty$: the model exhibits the
long-range order.
Notice that by (\ref{M}), the asymptotics (\ref{D1}) imply the existence of the
spontaneous magnetization for $T<T_c$ and the famous power-law decay
of the magnetization $M\sim \mathrm{Const}\cdot (T_c-T)^{1/8}$ as $T\nearrow T_c$.

At $T=T_c$, we have $t=0$, and therefore a Fisher-Hartwig singularity with the parameters
$\al=0$, $\bt=-1/2$ appears at $z=1$. In this case, (\ref{FH}) holds and we
obtain (cf. \cite{MW,McCoy})
\be\label{D2}
D_n(0)=\frac{\sqrt{\pi}G(1/2)^2}{n^{1/4}}(1+o(1)),
\ee
so the correlations decrease as $n^{-1/4}$, $n\to\infty$: the long-range order is
destroyed.
Note that as $V(z)\equiv 0$, there is an explicit expression for $D_n(0)$ (cf. the remark
following (\ref{FH},\ref{ab})):
\[
D_n(0)=\left({2\over n}\right)^n \prod_{k=1}^{n-1}\left(1-{1\over 4k^2}\right)^{k-n}.
\]

For $T>T_c$ we have $t<0$. The symbol still has a
singularity at $z=1$ but now with the parameters $\al=0$, $\bt=-1$.
This is the situation of a degenerate type of a Fisher-Hartwig singularity, and (\ref{FH})
does not hold in this case as one of the $G$-functions vanishes.
Calculations \cite{MW} show an
$n^{-1/2}e^{nt}$ decay of the correlations (\ref{corr}) as  $n\to\infty$.
There is no long-range order and $M=0$.

The transition $T\to T_c$ for large $n$ was studied in \cite{WMTB,
MTW, T} (for a more general correlation function
$<\si_{00}\si_{nm}>$). The authors took $n\to\infty$ with
$x=n(e^{2t}-1)$ {\it fixed} and found, in particular,
that in this limit
\be\label{Wu}
n^{1/4}<\si_{00}\si_{nn}>\rightarrow F(x),
\ee
where $F(x)$ is given in terms of a solution to Painlev\'e III equation (reducible
to Painlev\'e V: see equation (\ref{sigma5ising}) below). Moreover, in \cite{MTW},
McCoy, Tracy and Wu
evaluated the connection formulae for this Painlev\'e III function
and showed that the limiting behavior of $F(x)$ as $x \to \infty$ and as $x \to 0$
formally matches  the  asymptotics  (\ref{D1}) and  (\ref{D2}), respectively. The
matching with (\ref{D2}) was, however, checked only up to a multiplicative constant.
A more detailed evaluation of the small $x$ behavior of the function
$F(x)$ was carried out later  by Tracy in \cite{T}, and reproduced exactly
the constant in the critical point  asymptotics (\ref{D2}). The calculations
of \cite{MTW} and \cite{T} were  based on an alternative representation
of the function $F(x)$ as an infinite series of integrals and were rather
involved.

The results of the present paper, namely a particular case of Theorem \ref{theorem:
Toeplitz2}
below, fully describe (up to a possibility of a finite number of
poles: see below) the transition of
the correlation function (\ref{corr}) for large $n$ from $T<T_c$ to $T=T_c$,
that is the transition from (\ref{D1}) to (\ref{D2}).
Note that we do not rely on (\ref{Wu}) as the parameter
$x$ in our analysis is not necessarily fixed, in fact,
our asymptotics for the correlation function are {\it uniform}
in the whole range $x\in[0,\infty)$ (away from a finite number of positive points).
More precisely, our asymptotics as $n\to\infty$ are uniform for all
$T\in [T_1,T_c]$  (away from a finite number of positive $x$'s)
for some $T_1<T_c$ with $T_c-T_1$ sufficiently small.

The description of the Ising double scaling theory which we obtain as
a particular case of Theorem \ref{theorem: Toeplitz2} is in
agreement with the  classical results of \cite{WMTB, MTW, T} (see also
Remark \ref{MTWBus} below). Furthermore,
by obtaining the uniform asymptotics for the whole transition range of temperatures $T\le T_c$,
we complement the analysis of this case of the Wu-McCoy-Tracy-Barouch scaling theory.

Another example where Theorems \ref{theorem: Toeplitz}, \ref{theorem: Toeplitz2}
below can be applied is the so-called emptiness
formation probability in a Heisenberg spin chain \cite{Abanov}.

\subsection{Statement of results}
Consider the second order ODE
\begin{multline}
\left(x\frac{d^2\sigma}{dx^2}\right)^2 = \left(\sigma
-x\frac{d\sigma}{dx} +2\left(\frac{d\sigma}{dx}\right)^2 +
2\alpha\frac{d\sigma}{dx}\right)^2 \\
\label{sigma5}
-4\left(\frac{d\sigma}{dx}\right)^2\left(\frac{d\sigma}{dx} +\alpha
+\beta\right) \left(\frac{d\sigma}{dx} +\alpha -\beta\right).
\end{multline}
This is the Jimbo-Miwa-Okamoto $\sigma$-form \cite{J, jm2} of the
fifth Painlev\'e equation
\begin{equation}\label{PVintro}
u_{xx}=\left(\frac{1}{2u}+\frac{1}{u-1}\right)u_{x}^2-\frac{1}{x}u_x+\frac{(u-1)^2}{x^2}\left(A
u+\frac{B}{u}\right)+\frac{C u}{x} +D\frac{u(u+1)}{u-1},
\end{equation}
with the parameters $A, B, C, D$ given by
\begin{equation}\label{ABCD}
A=\frac{1}{2}(\alpha-\beta)^2,\qquad
B=-\frac{1}{2}(\alpha+\beta)^2,\qquad C=1+2\beta, \qquad
D=-\frac{1}{2}.
\end{equation}

\medskip

In the following theorem, we give the asymptotic expansion for the Toeplitz determinant
with symbol (\ref{symbol}) as $n\to \infty$ which is valid uniformly for $0<t<t_0$.
Our asymptotic expansion interpolates between Szeg\H{o} and Fisher-Hartwig
asymptotics.

\begin{theorem}\label{theorem: Toeplitz}
Let $\alpha\in \mathbb R$, $\alpha>-\frac{1}{2}$, $\beta\in i\mathbb
R$. Let $f$ be defined by (\ref{symbol}) and consider the Toeplitz
determinant $D_n(t)$ defined by (\ref{Toeplitz}) corresponding to
this symbol. The following asymptotic expansion holds as
$n\to\infty$ with the error term $o(1)$ uniform for $0< t <
t_0$ where $t_0$ is sufficiently small:
\begin{multline}\label{expansion Dn}
\ln D_n(t)=n V_0+(\al+\bt)nt+
\sum_{k=1}^{\infty}k\left[V_k-(\alpha+\beta)\frac{e^{-tk}}{k}\right]
\left[V_{-k}-(\alpha-\beta)\frac{e^{-tk}}{k}\right]\\
+\ln\frac{G(1+\alpha+\beta)G(1+\alpha-\beta)}{G(1+2\alpha)}
+\Om(2nt)+o(1),
\end{multline}
where $G(z)$ is Barnes' G-function, and
\begin{equation}
\label{def
Omega}\Om(2nt)=\int_0^{2nt}\frac{\sigma(x)-\alpha^2+\beta^2}{x}dx+(\alpha^2-\beta^2)\ln
2nt.
\end{equation}
The function $\sigma(x)$ is a particular solution to the equation
(\ref{sigma5}) which is real analytic on $(0, +\infty)$, and has the
following asymptotics for $x>0$: \be\label{westintro}
\sigma(x)=\begin{cases}\al^2-\bt^2+
\frac{\alpha^2-\beta^2}{2\alpha}\{x-x^{1+2\al}C(\alpha,
\beta)\}(1+\bigO(x)),& x\to 0,\quad 2\al\notin\mathbb Z\cr
\al^2-\bt^2+\bigO(x)+\bigO(x^{1+2\al})+\bigO(x^{1+2\al}\ln x),& x\to
0,\quad 2\al\in\mathbb Z\cr x^{-1+2\alpha}e^{-x}
\frac{-1}{\Gamma(\alpha-\beta)\Gamma(\alpha+\beta)} \left(1 +
\bigO\left(\frac{1}{x}\right)\right),& x\to +\infty,
\end{cases}
\ee
with
\be\label{Cab}
C(\al,\bt)=
\frac{\Gamma(1+\al+\bt)\Gamma(1+\al-\bt)}{\Gamma(1-\al+\bt)\Gamma(1-\al-\bt)}
\frac{\Gamma(1-2\al)}{\Gamma(1+2\al)^2}\frac{1}{1+2\al},
\ee
where $\Gamma(z)$ is Euler's $\Gamma$-function.
\end{theorem}

\begin{remark}
Later on, we will construct $\sigma(x)$ explicitly in terms of a
Riemann-Hilbert problem.
\end{remark}

\begin{remark}
With increasing effort, one can calculate more terms in the expansion
(\ref{expansion Dn}) using our approach.
\end{remark}

The function $\sigma=\sigma(x;\alpha,\beta)$ is defined for  $x\in
\mathbb C$ with a cut from zero to infinity. It is analytic in the
cut plane apart from possible poles. Asymptotics (\ref{westintro})
imply that there are no poles for $x$ positive and sufficiently
large. Hence the number of possible poles of $\sigma(x)$ on
$(0,+\infty)$ is finite. We show below that for $\al>-\frac{1}{2}$
real, $\bt$ imaginary, there are no poles on the real half-axis
$(0,+\infty)$. Therefore we took the intervals of the real line as a
path of integration in (\ref{def Omega}). For $\bt$ arbitrary,
$\Re\al>-1/2$, a similar result holds, but we have to choose a
path of integration in the complex plane avoiding possible poles
which we denote $\{x_1,\dots,x_\ell\}$. Namely, we have
\begin{theorem}\label{theorem: Toeplitz2}
Let $\alpha, \beta\in\mathbb C$ with $\Re\alpha>-\frac{1}{2}$,
$\al\pm\bt\neq -1,-2,\dots$, and let $s_\de$ denote a sector
$-\pi/2+\de<\arg x<\pi/2-\de$, $0<\de<\pi/2$. Let $f$ be defined by
(\ref{symbol}) and consider the Toeplitz determinants $D_n(t)$
defined by (\ref{Toeplitz}) corresponding to this symbol. There
exists a finite set $\{x_1, \ldots , x_{\ell}\}\in s_\de$ (with
$\ell=\ell(\alpha,\beta,\delta)$ and $x_j=x_j(\alpha,\beta)\neq 0$)
such that the expansion (\ref{expansion Dn}) holds uniformly for
$t\in s_\delta, |t|<t_0$ (with $t_0$ sufficiently small) as long as
$2nt$ remains bounded away from the set $\{x_1, \ldots ,
x_{\ell}\}$. The function $\Omega$ is defined by (\ref{def Omega}),
where the path of integration is chosen in $s_\delta$, connecting
$0$ with $2nt$ and not containing any of the points $\{x_1, \ldots ,
x_{\ell}\}$. Moreover $\sigma(x)$ solves the ODE (\ref{sigma5}) and has
the asymptotics in the mentioned sector given by (\ref{westintro}).
\end{theorem}

\begin{remark}
It follows from the representation (\ref{expansion Dn}) that the
residue of $\frac{1}{x}\sigma(x)$ at each of its poles in the sector
$-\pi/2<\arg x<\pi/2$ is an entire number. Different choices of the
integration contour in (\ref{def Omega}) correspond, in general, to
different branches of $\ln D_n(t)$. If $\alpha$ and $\beta$ are such
that $\frac{1}{x}\sigma(x)$ has a pole $x$, the determinant $D_n(t)$
is zero at $2nt=x$ (up to an $\bigO(1/n)$ error term).
\end{remark}

\begin{remark}
From now on, we will always consider $t\geq 0$ for simplicity. The
extension to $t\in s_\delta$ is straightforward.
\end{remark}

\begin{remark}\label{MTWBus}
In the example of the Ising model discussed in the
previous section, we have $\al=0$, $\bt=-1/2$ (see \ref{Iab}), and
equation (\ref{sigma5}) becomes
\begin{equation}\label{sigma5ising}
\left(x\frac{d^2\sigma}{dx^2}\right)^2 = \left(\sigma
-x\frac{d\sigma}{dx} +2\left(\frac{d\sigma}{dx}\right)^2 \right)^2 -
4\left(\frac{d\sigma}{dx}\right)^2\left(\left(\frac{d\sigma}{dx}\right)^2
- \frac{1}{4} \right).
\end{equation}
This is exactly the equation which was obtained in \cite{jm2}  for
the function
\begin{equation}\label{zetaF}
\zeta(x) \equiv x\frac{d}{dx}\ln F(x) - \frac{1}{4},
\end{equation}
where $F(x)$ is the right hand side of the Ising double scaling
limit (\ref{Wu}) (see also equation (4.16), with $r = x/2$,  in
\cite{McCoy}).
It follows immediately from our main result (\ref{expansion Dn})
that the function $\zeta(x)$ in \cite{jm2} and our function
$\sigma(x)$ coincide:
\begin{equation}\label{sigmazeta}
\zeta(x) = \sigma(x).
\end{equation}
Thus the application of our Theorem \ref{theorem: Toeplitz2} to the Toeplitz
determinant (\ref{Ising})  yields the complete analysis
(up to the question of existence of a finite number of positive poles $x_j$)
for $T\le T_c$ of
the Jimbo-Miwa Painlev\'e V version of the Wu-McCoy-Tracy-Barouch scaling
theory for the 2D Ising model.
\end{remark}

From the expansion (\ref{expansion Dn}), we can recover the Fisher-Hartwig
asymptotics for $\ln D_n(0)$.
Let $t\to 0$, and $n$ fixed in (\ref{expansion Dn}). Then, using (\ref{def Omega})
and (\ref{westintro}), we obtain
that $\Om(2nt)=(\al^2-\bt^2)\ln(2nt)+o(1)$ if $\Re\alpha>-\frac{1}{2}$. Substituting
this into (\ref{expansion Dn})
and recalling (\ref{logid}), we
obtain (\ref{FH}).

The expansion (\ref{expansion Dn}) should also be consistent with
the Szeg\H o asymptotics for $t$ fixed. We see immediately that the
$\bigO(n)$ term gives, for a fixed $t$, the corresponding term in
the Szeg\H o asymptotics. Consistency of the $\bigO(1)$ terms,
however, yields an interesting identity involving the Painlev\'e
function $\sigma(x)$ via (\ref{def Omega}):
\be
\Om(+\infty)=-\ln\frac{G(1+\alpha+\beta)G(1+\alpha-\beta)}{G(1+2\alpha)}.
\ee

\subsection{The Painlev\'e V Riemann-Hilbert problem}\label{PVsection}

We can say more about the function $\sigma(x)$ than we did in Theorem
\ref{theorem: Toeplitz}: we can construct it explicitly in terms of
a Riemann-Hilbert (RH) problem. Consider the contour
$\Gamma=\cup_{j=1}^6\Gamma_j$ in the complex plane (see Figure
\ref{figure: Gamma}), with
    \begin{align*}
    &\Gamma_1=\frac{1}{2}+e^{i\frac{\pi}{4}}\mathbb R^+,
\qquad\Gamma_2=\frac{1}{2}+e^{i\frac{3\pi}{4}}\mathbb R^+,
\qquad\Gamma_3=\frac{1}{2}+e^{i\frac{5\pi}{4}}\mathbb R^+,\\
&\Gamma_4=\frac{1}{2}+e^{i\frac{7\pi}{4}}\mathbb R^+,\qquad
\Gamma_5=(1,+\infty),\qquad\Gamma_6=(0,1),
    \end{align*}
    with $\Gamma_1, \ldots , \Gamma_5$ oriented towards infinity and $\Gamma_6$
oriented to the right.
By a standard convention, the ``+'' side of the curve is on the left as one faces the direction
of the curve's orientation.

Let $\Re\alpha>-\frac{1}{2}$
and consider the following RH problem for $\Psi=\Psi(\zeta;x,\alpha,\beta)$.
    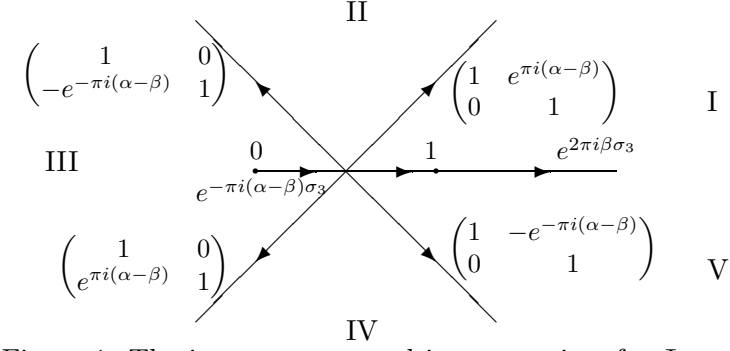
\begin{figure}[t]
\begin{center}
    \setlength{\unitlength}{0.8truemm}
    \begin{picture}(100,48.5)(0,2.5)

    \put(35,25){\thicklines\circle*{.8}}
    \put(34,27){\small $0$}
    \put(63,27){\small $1$}
    \put(65,25){\thicklines\circle*{.8}}
    \put(35,25){\line(1,0){30}}
    \put(45,25){\thicklines\vector(1,0){.0001}}
    \put(61,25){\thicklines\vector(1,0){.0001}}
    \put(50,25){\line(1,1){25}}
    \put(50,25){\line(1,-1){25}}
    \put(50,25){\line(-1,1){25}}
    \put(50,25){\line(-1,-1){25}}
    \put(65,25){\line(1,0){30}}
    \put(84,25){\thicklines\vector(1,0){.0001}}
    \put(65,40){\thicklines\vector(1,1){.0001}}
    \put(65,10){\thicklines\vector(1,-1){.0001}}
    \put(35,40){\thicklines\vector(-1,1){.0001}}
    \put(35,10){\thicklines\vector(-1,-1){.0001}}

    \put(67,37){\small $\begin{pmatrix}1&e^{\pi i(\alpha-\beta)}\\0&1\end{pmatrix}$}
    \put(-4,40){\small $\begin{pmatrix}1&0\\-e^{-\pi i(\alpha-\beta)}&1\end{pmatrix}$}
    \put(85,27){\small $e^{2\pi i\beta\sigma_3}$}
    \put(2,8){\small $\begin{pmatrix}1&0\\e^{\pi i(\alpha-\beta)}&1\end{pmatrix}$}
    \put(67,11){\small $\begin{pmatrix}1&-e^{-\pi i(\alpha-\beta)}\\0&1\end{pmatrix}$}
    \put(25,20){\small $e^{-\pi i(\alpha-\beta)\sigma_3}$}
    \put(110,35){I}
    \put(50,50){II}
    \put(0,25){III}
    \put(110,7){V}
    \put(50,-3){IV}
    \end{picture}
    \caption{The jump contour and jump matrices for $\Psi$.}
    \label{figure: Gamma}
\end{center}
\end{figure}
    \subsubsection*{RH problem for $\Psi$}
    \begin{itemize}
    \item[(a)] $\Psi:\mathbb C\setminus \Gamma \to \mathbb C^{2\times 2}$ is analytic.
    \item[(b)] $\Psi$ has continuous boundary values on
$\Gamma\setminus\left\{0,\frac{1}{2},1\right\}$, and they are related as
follows,
    \begin{align}
    &\label{RHP Psi:b1}\Psi_+(\zeta)=\Psi_-(\zeta)\begin{pmatrix}1&e^{\pi
i(\alpha-\beta)}\\0&1\end{pmatrix}, &\mbox{ for $\zeta\in\Gamma_1$,}\\
    &\Psi_+(\zeta)=\Psi_-(\zeta)\begin{pmatrix}1&0\\-e^{-\pi
i(\alpha-\beta)}&1\end{pmatrix}, &\mbox{ for $\zeta\in\Gamma_2$,}\\
    &\Psi_+(\zeta)=\Psi_-(\zeta)\begin{pmatrix}1&0\\e^{\pi
i(\alpha-\beta)}&1\end{pmatrix}, &\mbox{ for $\zeta\in\Gamma_3$,}\\
    &\Psi_+(\zeta)=\Psi_-(\zeta)\begin{pmatrix}1&-e^{-\pi
i(\alpha-\beta)}\\0&1\end{pmatrix}, &\mbox{ for $\zeta\in\Gamma_4$,}\\
&\Psi_+(\zeta)=\Psi_-(\zeta)e^{2\pi i\beta\sigma_3}, &\mbox{ for $\zeta\in\Gamma_5$},\\
    &\label{RHP Psi:b6}\Psi_+(\zeta)=\Psi_-(\zeta)e^{-\pi i(\alpha-\beta)\sigma_3},
&\mbox{ for $\zeta\in\Gamma_6$},
    \end{align}
    with $\sigma_3=\begin{pmatrix}1&0\\0&-1\end{pmatrix}$.
    \item[(c)] $\Psi$ has the following behavior as $\zeta\to\infty$
(for some matrices $C_1=C_1(x,\alpha,\beta)$, $C_2=C_2(x,\alpha,\beta)$),
    \begin{equation}\label{RHP Psi: c}
    \Psi(\zeta)=\left(I+\frac{C_1}{\zeta}+\frac{C_2}{\zeta^2}+\bigO(\zeta^{-3})\right)
\zeta^{-\beta\sigma_3}e^{-\frac{x}{2}\zeta\sigma_3}.
    \end{equation}
\item[(d0)]As $\zeta\to 0$,
    \begin{equation}\label{RHP Psi: d0}
    \Psi(\zeta)=\bigO\begin{pmatrix}|\zeta|^{\frac{\alpha-\beta}{2}}&|\zeta|^{-\frac{\alpha-\beta}{2}}\\
    |\zeta|^{\frac{\alpha-\beta}{2}}&|\zeta|^{-\frac{\alpha-\beta}{2}}\end{pmatrix}.
    \end{equation}
    \item[(d1)] As $\zeta\to 1$,
    \begin{equation}\label{RHP Psi: d1}
    \Psi(\zeta)=\bigO\begin{pmatrix}|\zeta-1|^{-\frac{\alpha+\beta}{2}}&|\zeta-1|^{\frac{\alpha+\beta}{2}}\\
    |\zeta-1|^{-\frac{\alpha+\beta}{2}}&|\zeta-1|^{\frac{\alpha+\beta}{2}}\end{pmatrix}.
    \end{equation}
    Furthermore $\Psi$ is bounded near $\frac{1}{2}$.
    \end{itemize}
    The RH conditions imply (by a standard argument) that the determinant of
the solution $\Psi$ (which is, if it exists, unique) is identically equal to
$1$, and consequently we have using (\ref{RHP Psi: c}) that $\tr C_1=0$.
Let us denote the matrix elements of $C_1$ by
\[C_1(x)=\begin{pmatrix}q(x)&r(x)\\t(x)&-q(x)\end{pmatrix}.\]
Define the functions $v$ and $u$ in terms of the matrix elements of $C_1$:
\begin{align}
&\label{def v intro}v(x)=\frac{\alpha+\beta}{2}-q(x)-xr(x)t(x), \\
&\label{def u intro}u(x)=1+\frac{xt}{(2\beta+1-x)t(x)+xt'(x)}.
\end{align}
We will show in Section \ref{section 43} below that
\begin{equation}\label{def sigma}
\sigma(x)=\int_x^{+\infty}v(\xi)d\xi
\end{equation}
is the function appearing in Theorem \ref{theorem: Toeplitz} and
Theorem \ref{theorem: Toeplitz2}. The RH problem for $\Psi$ is a
special case of the RH problem associated to the fifth Painlev\'e
equation, see e.g.\ \cite{FIKN, FokasMuganZhou}.

\medskip

We prove the following.
\begin{theorem}\label{theorem: Painleve}
Let $\alpha, \beta \in\mathbb C$ and $\Re\alpha>-\frac{1}{2}$.
\begin{itemize}
\item[(i)]
The RH problem for $\Psi$ is uniquely solvable for all $x>0$
except possibly for a finite number of positive $x$-values.
We denote the $x$-values for which the RH problem
is not solvable by $\{x_1, \ldots , x_k\}$, with $x_j=x_j(\alpha, \beta)$ and
$k=k(\alpha, \beta)$.
\item[(ii)] If $\Im\alpha=0$ and $\Re\beta=0$ the RH problem is solvable for all
positive $x$-values.
\item[(iii)]
The function $v$ defined by (\ref{def v intro}) is analytic in
$(0,+\infty)\setminus\{x_1, \ldots, x_k\}$, and solves, together
with $u$ defined by (\ref{def u intro}), the system
\begin{align}\label{system uv introb}
&xu_x=xu-2v(u-1)^2+(u-1)[(\alpha-\beta)u-\beta-\alpha],\\
&\label{system uv2
introb}xv_x=uv[v-\alpha+\beta]-\frac{v}{u}(v-\beta-\alpha).
\end{align}
 \item[(iv)] The function
$v$ has the asymptotics given by \be\label{vestintro} v(x)=\begin{cases}
-\frac{\alpha^2-\beta^2}{2\alpha}\{1-(2\alpha+1)x^{2\al}C(\alpha,
\beta)\}(1+\bigO(x)),& x\to 0,\quad 2\al\notin\mathbb Z,\cr
\bigO(1)+\bigO(x^{2\al})+\bigO(x^{2\al}\ln x),& x\to 0,\quad
2\al\in\mathbb Z\cr x^{-1+2\alpha}e^{-x}
\frac{-1}{\Gamma(\alpha-\beta)\Gamma(\alpha+\beta)} \left(1 +
\bigO\left(\frac{1}{x}\right)\right),& x\to +\infty,
\end{cases}
\ee
where $C(\al,\bt)$ is defined in (\ref{Cab}).

In addition, we have
\begin{equation}\label{intv}
\int_0^{+\infty}v(x)dx=\alpha^2-\beta^2
\end{equation}
if $\Im\al=0$ and $\Re\bt=0$. In the general case,  $\alpha, \beta \in\mathbb C$,
$\Re\alpha>-\frac{1}{2}$, equation (\ref{intv}) holds up to addition of $2\pi i m$,
$m\in\mathbb Z$, with the path of integration avoiding $\{x_1, \ldots, x_k\}$.
\end{itemize}
\end{theorem}

Part (iii) of the theorem follows from a standard Lax pair argument and was proved in
\cite{FokasMuganZhou, FIKN} for a slightly different but equivalent RH
problem. That proof applies to our RH problem as well, and implies moreover that
the RH solution is meromorphic in $x$ for $x\in\mathbb C\setminus\{0\}$. We will
come back to this in Section \ref{section: Painleve}.
We
prove part (iv) by performing the Deift-Zhou steepest descent analysis
for the RH problem. This asymptotic analysis also implies the
solvability of the RH problem for large $x$ and small $x$, and by meromorphicity in
$x$, this leads
to the statement (i).
We prove part (ii) by applying the technique of a vanishing lemma to
the RH problem for $\Psi$.

\begin{remark}\label{MTWBus2}
The system (\ref{system uv introb})-(\ref{system uv2 introb}) is
related to the Painlev\'e V equation: eliminating $v$, we easily
verify that $u$ solves the Painlev\'e V equation
(\ref{PVintro})-(\ref{ABCD}). Asymptotic expansions as $x\to 0$ and
as $x\to\infty$ for various solutions to the
fifth Painlev\'e equation and the system (\ref{system uv
introb})-(\ref{system uv2 introb}) were obtained in several
works, see e.g.\ \cite{Andreev, AndreevKitaev, AndreevKitaev2, J,
MT, Shukla}. The solution $v$ which is of interest to us decays
exponentially at $+\infty$, is integrable near $0$ if
$\Re\alpha>-\frac{1}{2}$, and it has no poles on $(0,+\infty)$ if
$\alpha>-\frac{1}{2}\in\mathbb R$ and $\beta\in i\mathbb R$.
Note that the asymptotics of $\sigma(x)$ (\ref{westintro}) follow
from (\ref{vestintro}) and (\ref{intv}) by (\ref{def sigma}).
Although we obtained the connection formulae (\ref{westintro}) as a by-product of our analysis,
we would expect that these asymptotics  can be found  in the general list of
connection formulae for the fifth Painlev\'e equation obtained in
\cite{Andreev, AndreevKitaev, AndreevKitaev2}.

\end{remark}

\subsubsection*{Outline of the paper} The proofs of Theorem \ref{theorem: Toeplitz}
and Theorem \ref{theorem: Toeplitz2} are based on a well-known connection between
Toeplitz determinants and orthogonal polynomials on the unit circle. In Section
\ref{section: OP}, we obtain a
differential identity for $\ln D_n(t)$ in terms of the polynomials orthogonal on the
unit circle with weight $f(z)$.
In Section \ref{section: RH}, we obtain large $n$ asymptotics
for these orthogonal polynomials from a RH problem. The asymptotics will be given in
terms of a model RH
problem which we study in detail in Section \ref{section:
Painleve}, where we also give a proof of Theorem \ref{theorem: Painleve}.
In Section \ref{section: integrate}, we use the previously obtained
asymptotics for the orthogonal polynomials and the results of Section
\ref{section: Painleve} to integrate the differential identity for
$\ln D_n(t)$, which leads to Theorem \ref{theorem: Toeplitz} and Theorem
\ref{theorem: Toeplitz2}.

\medskip

Throughout the paper, we choose the branches of logarithms and roots corresponding
to arguments between $0$ and $2\pi$, unless stated otherwise.

\section{RH problem for orthogonal polynomials and a differential identity
for the Toeplitz determinants}
\label{section: OP}

Our analysis is based on a classical connection between Toeplitz determinants and
orthogonal polynomials.
Assume that for some $n>0$ $D_n,D_{n+1}\neq 0$, and
define a polynomial $\phi_n(z)$ in terms of the Fourier coefficients of $f(z)$ as
follows:
\begin{equation}\label{ef1}
\phi_n(z)={\frac{1}{\sqrt{D_n D_{n+1}}}}
\left|
\begin{matrix}
f_0& f_{-1}& \cdots & f_{-n}\cr
f_1& f_0& \cdots & f_{-n+1}\cr
\vdots & \vdots &  & \vdots \cr
f_{n-1} & f_{n-2} & \cdots & f_{-1} \cr
1& z& \cdots & z^n
\end{matrix}
\right|.
\end{equation}
The leading coefficient of $\phi_n$ is then equal to
\begin{equation}\label{chi}
\chi_n=\sqrt{\frac{D_n}{D_{n+1}}}.
\end{equation}
There holds the orthogonality relation
\begin{equation}\label{or1}
\frac{1}{2\pi}\int_{C}\phi_n(z)z^{-j}f(z)\frac{dz}{iz}=\chi_n^{-1}\delta_{jn},\qquad
j=0,1,\dots n,
\end{equation}
where $C$ is the unit circle oriented in the counterclockwise direction.
Similarly, let $\hat\phi_n(z)$ be defined by
\begin{equation}\label{ef2}
\hat\phi_n(z)={1\over\sqrt{D_n D_{n+1}}}
\left|
\begin{matrix}
f_0& f_{-1}& \cdots &f_{-n+1}& 1\cr
f_1& f_0& \cdots &f_{-n+2}& z\cr
\vdots & \vdots &  & \vdots \cr
f_n& f_{n-1}& \cdots &f_1& z^n
\end{matrix}
\right|.
\end{equation}
Then $\hat\phi_n$ has the same leading coefficient $\chi_n$ as $\phi_n$, and
\begin{equation}\label{or2}
{1\over 2\pi}\int_{C}\hat\phi_n(z^{-1})z^j f(z)\frac{dz}{iz}=
\chi_n^{-1}\delta_{jn},\qquad j=0,1,\dots,n.
\end{equation}
If $D_n\neq 0$ for $n=1,\ldots $ (and we set $D_0\equiv 1$,
$\phi_0(z)=\hat\phi_0(z)=1/\sqrt{D_1}$),
the system of polynomials $\phi_n$ and $\hat\phi_n$, $n=0,1,\dots$ exists
and can be characterized by the orthonormality relations
    \begin{equation}
    \frac{1}{2\pi}\int_{C}\phi_k(z)\hat\phi_m(z^{-1})f(z)\frac{dz}{iz}=\delta_{km},
\qquad k,m=0,1,\ldots.
    \end{equation}

If the symbol $f$ is positive on the unit circle $C$,
it is a classical fact (which follows, e.g., from the representation of a Toeplitz
determinant
as a multiple integral) that $D_n(f)>0$ for all $n\geq 0$,
and the system of orthogonal polynomials exists.

Assume that $D_{n-1}, D_n, D_{n+1}\neq 0$, $t>0$, and
define the function $Y(z;n)$ as follows
    \begin{equation}\label{def Y}
    Y(z)=
    \begin{pmatrix}
    \chi_n^{-1}\phi_n(z)&\chi_n^{-1}\int_{C}\frac{\phi_n(\xi)}{\xi-z}\frac{f(\xi)d\xi}{2\pi
i\xi^n}\\
    -\chi_{n-1}z^{n-1}\hat\phi_{n-1}(z^{-1})&-\chi_{n-1}\int_{C}\frac{\hat\phi_{n-1}(\xi^{-1})}{\xi-z}
    \frac{f(\xi)d\xi}{2\pi i\xi}
    \end{pmatrix}.
    \end{equation}
Then $Y$ is the unique solution of the following RH problem with
a jump on the counterclockwise oriented unit circle $C$.
    \subsubsection*{RH problem for $Y$}
    \begin{itemize}
    \item[(a)] $Y:\mathbb C \setminus C \to \mathbb C^{2\times 2}$ is analytic.
    \item[(b)] $Y_+(z)=Y_-(z)
                \begin{pmatrix}
                    1 & z^{-n}f(z) \\
                    0 & 1
                \end{pmatrix},$
                \qquad  for $z \in C$.
    \item[(c)] $Y(z)=\left(I+\bigO(1/z)\right)
                \begin{pmatrix}
                    z^n & 0 \\
                    0 & z^{-n}
                \end{pmatrix}$,
                \qquad  as $z\to \infty $.
    \end{itemize}

\medskip

A general fact that orthogonal polynomials can be so represented as a solution
of a RH problem was noticed in \cite{FIK} (for polynomials
on the line) and extended for polynomials on the circle in \cite{BDJ}.

In the next section we will show that the RH problem for $Y(z;n,t)$
is solvable (and therefore the orthogonal polynomials exist and the
coefficients $\chi_n$ are nonzero) for all $n$ larger than some
$n_0(\al,\bt)$ provided $2nt$ is bounded away from a certain finite
set of points (in particular, see Proposition \ref{prop21} below, $D_n\neq 0$). The number
$n_0(\al,\bt)$ is bounded for $\al$ and $\bt$ in a bounded set.

Our next aim is to express $\frac{d}{dt}\ln D_n(t)$ in terms of the entries of the RH
solution $Y$. We prove the following.
\begin{proposition}\label{prop21}
Let $t>0$ and $n\in\mathbb N$. Suppose that the RH problem for $Y(z;n,t)$ is solvable.
Then $D_n\neq 0$, and the following differential identity holds:
\be\label{differentialidentity}
\frac{d}{dt}\ln D_n(t)=-(\al+\bt)e^t \left(Y^{-1}{dY\over dz}\right)_{22}(e^t)+
(\al-\bt)e^{-t} \left(Y^{-1}{dY\over dz}\right)_{22}(e^{-t}).
\ee
\end{proposition}

\begin{proof}
We will follow the approach of \cite{ITW2}.
Let us start with the expression
\be\label{DnKn}
D_n(f)=\det(I-K_n),
\ee
where $K_n$ is an integral operator acting on $L^2(C)$ with kernel
\be
K_n(z,z')=\frac{(z/z')^n-1}{z-z'}\frac{1-f(z')}{2\pi i}.
\ee
This fact is easy to verify by considering the matrix expression for
$K_n$ in the basis $\{ z^k \}$, $k=-\infty,\dots,\infty$.

As follows from the theory of ``integrable'' Fredholm
operators (see, e.g. \cite{ITW2}), the solvability of the RH problem
for $Y(z;n,t)$  implies that the operator $1-K_n$ is invertible.
Therefore, $D_n(f)=\det(I-K_n)\neq 0$.

Consider
\be
{d \over d t}\ln D_n(f)=
{d \over d t}\tr\ln (I-K_n)=
-\tr (I-K_n)^{-1}{d K_{n}\over d t}.
\ee
Since
\be
{d f\over d t}=
\left(-{\al+\bt\over z-e^t}e^t +{\al-\bt\over z-e^{-t}}e^{-t}\right)f,
\ee
we have
\be\label{Dtr}
{d K_n\over d t}=K_{n}^{(1)}-K_{n}^{(2)},
\ee
where
\be\label{k2}
K_{n}^{(2)}(z,z')={\al-\bt\over z'-e^{-t}}
\frac{(z/z')^n-1}{z-z'}\frac{f(z')}{2\pi i}e^{-t}
\ee
and
\be\label{k1}
K_{n}^{(1)}(z,z')=\Lambda_1(z,z')-
{\al+\bt\over z'-e^{t}}
\frac{(z/z')^n-1}{z-z'}\frac{1-f(z')}{2\pi i}e^t.
\ee
Here
\be
\Lambda_1(z,z')={\al+\bt\over z'-e^t}{e^t\over 2\pi i}\frac{(z/z')^n-1}{z-z'}.
\ee
The reason to single out $\Lambda_1$ will soon become clear.
By a residue calculation, we obtain
\[
(\Lambda_1 K_{n})(z,z')=
(\al+\bt)e^t\frac{1-f(z')}{2\pi i}\frac{1}{z-e^t}
\left[
\frac{(z/z')^n-1}{z-z'}+
\frac{1}{e^t-z'}\left(
\left({z\over e^t}\right)^n-\left({z\over z'}\right)^n\right)
\right].
\]
We can now rewrite (\ref{k1}) as follows:
\be\label{k11}
K_{n}^{(1)}(z,z')=(\Lambda_1(I-K_{n}))(z,z')+
(\al+\bt)e^t\frac{1-f(z')}{2\pi i}
\frac{1-(z/e^t)^n}{(z-e^t)(z'-e^t)}.
\ee
Defining the following 2-component vectors
\[
\hat f(z)={z^n \choose 1},\quad \hat g(z)={1-f(z)\over 2\pi i}{z^{-n}\choose -1},\quad
\wt f(z)={ \hat f(z)\over z-e^t},\quad
\wt g(z)={ \hat g(z)\over z-e^t},
\]
we can write (\ref{k11}) in the final form:
\be\label{k111}
K_{n}^{(1)}(z,z')=
(\Lambda_1(I-K_{n}))(z,z')+
(\al+\bt)e^t (e^{-nt}\wt f_1(z)\wt g_2(z')-\wt f_2(z)\wt g_2(z'))
\ee

On the other hand, let us define $F_j=(I-K_n)^{-1}\hat f_j$. Then
\[
{F_j(z)\over z-e^t}-
\int_C {K_n(z,z')\over z-e^t}F_j(z')dz'=\wt f_j(z),\qquad j=1,2.
\]
Noting that
\[
K_n(z,z')=\frac{\hat f(z)^T \hat g(z')}{z-z'},
\]
we can write the above equation in the form
\[
{F_j(z)\over z-e^t}-
\int_C K_n(z,z'){F_j(z')\over z'-e^t}dz'+
\int_C \wt f(z)^T \hat g(z'){F_j(z')\over z'-e^t}dz'=\wt f_j(z).
\]
Applying $(I-K_n)^{-1}$ to both sides, we obtain
\be
\sum_{k=1}^2 m_{jk}(e^t)\wt F_k(z)={1\over z-e^t}F_j(z),\qquad j=1,2,
\ee
where $\wt F_k=(I-K_n)^{-1}\wt f_k$ and
\be
m_{jk}(e^t)=\de_{jk}-\int_C F_j(z)\hat g_k(z){dz\over z-e^t},\qquad j,k=1,2.
\ee
Thus we have for the 2-component vector ($\det m=1$, see \cite{ITW2})
\be\label{mF}
\wt F(z)={1\over z-e^t}m^{-1}(e^t)F(z)=
{1\over z-e^t}
\begin{pmatrix}
m_{22}F_1-m_{12}F_2\cr
-m_{21}F_1+m_{11}F_2
\end{pmatrix}.
\ee
As is shown in (\cite{ITW2}, Eq. (2.16) up to a different notation),
the matrix $m$ is related to $Y$. For $|z|>1$,
\be\label{Ym+}
Y(z)=\begin{pmatrix}
m_{11}z^n+m_{12}& -m_{12}z^{-n}\cr
-m_{21}z^n-m_{22}& m_{22}z^{-n}
\end{pmatrix}.
\ee
Using the definition $\wt F=(I-K_n)^{-1}\wt f$, equations (\ref{k111}),
(\ref{mF}), (\ref{Ym+}), and the fact that
\[
\tr\Lambda_1=(\al+\bt){e^t\over2\pi i}\int_C\frac{n dz'}{z'(z'-e^t)}=
-n(\al+\bt),
\]
we easily obtain
\be\label{t1}
\tr ((I-K_n)^{-1}K_{n}^{(1)})=
(\al+\bt)e^t (Y_{11}(e^t){Y'_z}_{22}(e^t)-Y_{21}(e^t){Y'_z}_{12}(e^t)),
\ee
where $Y'_z(e^t)$ stands for the derivative of $Y(z)$ w.r.t. $z$ evaluated at
$z=e^t$.

Let us now compute the contribution of $K_{n}^{(2)}$.
First, write $K_{n}^{(2)}$ in the form
\[
K_{n}^{(2)}=\Lambda_2-
{\al-\bt\over z'-e^{-t}}
\frac{(z/z')^n-1}{z-z'}\frac{1-f(z')}{2\pi i}e^{-t},\qquad
\Lambda_2={\al-\bt\over z'-e^{-t}}
\frac{(z/z')^n-1}{z-z'}\frac{1}{2\pi i}e^{-t}.
\]
We then obtain as above for $K_{n}^{(1)}$ that
\be\label{k22}
K_{n}^{(2)}(z,z')=(\Lambda_2(I-K_{n}))(z,z')+
(\al-\bt)e^{-t}\frac{1-f(z')}{2\pi i}
\frac{(e^{-t}/z')^n-(z/z')^n}{(z-e^{-t})(z'-e^{-t})}.
\ee
Defining the new vectors
\[
\wt f(z)={ \hat f(z)\over z-e^{-t}},\quad
\wt g(z)={ \hat g(z)\over z-e^{-t}},\quad
\wt F_k=(I-K)^{-1}\wt f_k,
\]
we can write (\ref{k22}) in the form:
\be\label{k222}
K_{n}^{(2)}(z,z')=
(\Lambda_2(I-K_{n}))(z,z')+
(\al-\bt)e^{-t} (e^{-nt}\wt f_2(z)\wt g_1(z')-\wt f_1(z)\wt g_1(z'))
\ee
and obtain as above
\be\label{mF2}
\wt F(z)=
{1\over z-e^{-t}}
\begin{pmatrix}
m_{22}(e^{-t})F_1(z)-m_{12}(e^{-t})F_2(z)\cr
-m_{21}(e^{-t})F_1(z)+m_{11}(e^{-t})F_2(z)
\end{pmatrix}.
\ee
For $|z|<1$, the matrix $m$ is related to $Y$
by the formula preceeding (3.21) in \cite{ITW2}:
\be\label{Ym-}
Y(z)=\begin{pmatrix}
m_{11}z^n+m_{12}& m_{11}\cr
-m_{21}z^n-m_{22}& -m_{21}
\end{pmatrix}.
\ee
Using the definition $\wt F=(I-K_n)^{-1}\wt f$, equations (\ref{k222}),
(\ref{mF2}), (\ref{Ym-}), and the fact that $\tr\Lambda_2=0$,
we finally obtain
\be\label{t2}
\tr ((I-K_n)^{-1}K_{n}^{(2)})=
(\al-\bt)e^{-t} (Y_{11}(e^{-t}){Y'_z}_{22}(e^{-t})-Y_{21}(e^{-t}){Y'_z}_{12}(e^{-t})).
\ee

Expressions (\ref{t1}) and (\ref{t2}) imply by (\ref{Dtr}) the statement of the
Proposition.
\end{proof}

\section{Asymptotic analysis of the RH problem for orthogonal
polynomials}\label{section: RH}
In this section, we apply the steepest descent method of Deift and Zhou \cite{DZ} to
the RH
problem for $Y$. We follow the strategy used in \cite{BDJ} for orthogonal
polynomials on the unit circle with a special weight function, and used in
\cite{DKMVZ2, DKMVZ1} for orthogonal polynomials on the real line with respect to a
more general weight function. The most important new feature here is the
construction of a local parametrix near $1$ which involves a Painlev\'e V RH
problem.
We will obtain asymptotics for $Y$ as $n\to\infty$, and therefore,
asymptotics of the r.h.s.\ of (\ref{differentialidentity}) in terms of Painlev\'e V
functions.

    \subsection{Normalization of the RH problem}
    Define
    \begin{equation}\label{def T}
    T(z)=\begin{cases}Y(z)z^{-n\sigma_3},&\mbox{ as $|z|>1$},\\
    Y(z),&\mbox{ as $|z|<1$,}
    \end{cases}
    \end{equation}
    with $Y$ given by (\ref{def Y}).
    Then $T$ satisfies a RH problem normalized at infinity.
    \subsubsection*{RH problem for $T$}
    \begin{itemize}
    \item[(a)] $T:\mathbb C \setminus C \to \mathbb C^{2\times 2}$ is analytic.
    \item[(b)] $T_+(z)=T_-(z)
                \begin{pmatrix}
                    z^n & f(z) \\
                    0 & z^{-n}
                \end{pmatrix},$
                \qquad  for $z \in C$.
    \item[(c)] $T(z)=I+\bigO(1/z)$
                \qquad  as $z\to \infty $.
    \end{itemize}
The diagonal elements of
the jump matrix for $T$ oscillate rapidly on the unit circle if $n$ is large.
The next transformation turns the oscillatory behavior into exponential decay on a
deformed contour.

\subsection{Opening of the lens}

    Note that one can factorize the jump matrix for $T$ as follows,
    \begin{eqnarray}
    J_T(z)&:=&\begin{pmatrix}
                    z^n & f(z) \\
                    0 & z^{-n}
                \end{pmatrix}\nonumber \\
                &=&\begin{pmatrix}
                    1& 0 \\
                    z^{-n}f(z)^{-1} & 1
                \end{pmatrix}\begin{pmatrix}
                    0 & f(z) \\
                    -f(z)^{-1} & 0
                \end{pmatrix}\begin{pmatrix}
                    1 & 0 \\
                    z^nf(z)^{-1} & 1
                \end{pmatrix}.
    \end{eqnarray}

    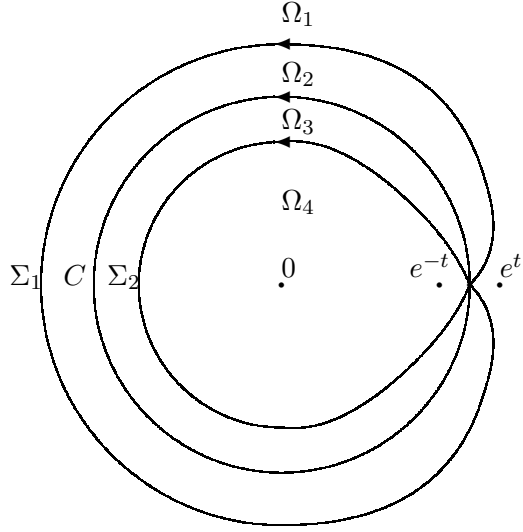
\begin{figure}[t]
    \begin{center}
    \setlength{\unitlength}{1truemm}
    \begin{picture}(100,55)(-5,10)
        \cCircle(50,40){25}[f]
        \put(50,40){\thicklines\circle*{.8}}
        \put(50,41){$0$}
        \put(75,40){\thicklines\circle*{.8}}
        \put(71,40){\thicklines\circle*{.8}}
        \put(79,41){$e^t$}
        \put(79,40){\thicklines\circle*{.8}}
        \put(67,41){$e^{-t}$}
        \put(49,65){\thicklines\vector(-1,0){.0001}}
        \put(49,72){\thicklines\vector(-1,0){.0001}}
        \put(49,59){\thicklines\vector(-1,0){.0001}}
        \put(21,40){$C$}
        \put(13.8,40){$\Sigma_1$}
        \put(26.8,40){$\Sigma_2$}

        \put(50,50){$\Omega_4$}
        \put(50,60.7){$\Omega_3$}
        \put(50,67){$\Omega_2$}
        \put(50,75){$\Omega_1$}
        \qbezier(75,40)(80,45)(77,54)
        \qbezier(77,54)(72,72)(50,72)
        \cCircle(50,40){32}[l]
\qbezier(75,40)(80,35)(77,26)
        \qbezier(77,26)(72,8)(50,8)

        \qbezier(75,40)(73,45)(68,50)
        \qbezier(68,50)(58,60)(50,59)
        \cCircle(50,40){19}[l]
        \qbezier(75,40)(73,35)(68,30)
        \qbezier(68,30)(58,20)(50,21)

    \end{picture}
    \caption{The contour $\Sigma=\Sigma_1\cup C\cup\Sigma_2$ and the regions
 $\Omega_1,\ldots,\Omega_4$.}
    \label{fig2}
\end{center}
\end{figure}

To make use of
this factorization, consider the
three counterclockwise oriented closed curves as shown in Figure
    \ref{fig2}. Let us write
    \begin{align}\label{J1}
    &J_1(z)=\begin{pmatrix}
                    1& 0 \\
                    z^{-n}f(z)^{-1} & 1
                \end{pmatrix},\\
    &\label{JN}J_N(z)=\begin{pmatrix}
                    0 & f(z) \\
                    -f(z)^{-1} & 0
                \end{pmatrix},\\
    &\label{J2}J_2(z)=\begin{pmatrix}
                    1 & 0 \\
                    z^nf(z)^{-1} & 1
                \end{pmatrix},
    \end{align}
    and
    define
    \begin{equation}\label{def S}
    S(z)=
    \begin{cases}
    T(z), &\mbox{ in regions $\Omega_1$ and $\Omega_4$},\\
    T(z)J_1, &\mbox{ in region $\Omega_2$},\\
    T(z)J_2^{-1}, &\mbox{ in region $\Omega_3$},
    \end{cases}
    \end{equation}
    with $\Omega_1, \Omega_2, \Omega_3$ as indicated in Figure \ref{fig2}.
Note that $f$ is an analytic function in $\mathbb C\setminus
([0,e^{-t}]\cup[e^{t},+\infty))$,
and we define $\Sigma_1$, $\Sigma_2$
so that the branch cuts are located in the regions $\Omega_1$ and $\Omega_4$. As we
can have $t\to 0$, it
is inevitable that $\Omega_1$ and $\Omega_4$ approach $1$. We choose $\Sigma_1$ and
$\Sigma_2$
to go through $1$ as in Figure \ref{fig2}.
The function $S(z)$ satisfies the following RH problem.
    \subsubsection*{RH problem for $S$}
    \begin{itemize}
    \item[(a)] $S:\mathbb C \setminus (\Sigma_1\cup C \cup\Sigma_2) \to \mathbb
C^{2\times 2}$ is analytic.
    \item[(b)] $S_+(z)=S_-(z)
                J_k(z),$
                \qquad  for $z \in \Sigma_k$, $k=1,2$,\\
$S_+(z)=S_-(z)J_N(z)$,\qquad for $z\in C$.
    \item[(c)] $S(z)=I+\bigO(1/z)$,
                \qquad  as $z\to \infty $.
    \end{itemize}
Note that the jump matrices $J_1$ and $J_2$ tend to the identity matrix on their
respective contours $\Sigma_1$ and $\Sigma_2$ as $n\to\infty$  except near $1$.

We need to construct a parametrix dealing with the jump condition on the unit circle
and a local parametrix near $1$.

    \subsection{Global parametrix away from $1$}
Ignoring the jumps on $\Sigma_1$ and $\Sigma_2$ and a neighborhood of $1$, we
consider the following model problem.
    \subsubsection*{RH problem for $N$}
    \begin{itemize}
    \item[(a)] $N:\mathbb C \setminus C \to \mathbb C^{2\times 2}$ is analytic.
    \item[(b)] $N_+(z)=N_-(z)
                J_N(z),$
                \qquad  for $z \in C$.
    \item[(c)] $N(z)=I+\bigO(1/z)$,
                \qquad  as $z\to \infty $.
    \end{itemize}
    This problem is easily solved explicitly:
    \begin{equation}\label{def N}
    N(z)=\begin{cases}
    D(z)^{\sigma_3}\begin{pmatrix}0&1\\-1&0\end{pmatrix},&\mbox{ for $|z|<1$},\\
    D(z)^{\sigma_3},&\mbox{ for $|z|>1$},
    \end{cases}
    \end{equation}
where the (Szeg\H o) function $D(z)$ is analytic
and nonzero in $\mathbb C\setminus C$, tends to $1$
as $z\to\infty$, and satisfies the jump condition $D_+(z)=D_-(z)f(z)$ for
$z\in C$. It is easy to verify that
    \begin{equation}
    D(z)=\begin{cases}\label{def D2}
    (z-e^{t})^{\alpha+\beta}e^{-i\pi(\alpha+\beta)}\exp\left(\sum_{k=0}^{\infty}V_k
z^{k}\right),&\mbox{ for $|z|<1$},\\
    (z-e^{-t})^{-\alpha+\beta}z^{\alpha-\beta}\exp\left(-\sum_{k=-\infty}^{-1}V_k
z^{k}\right),&\mbox{ for $|z|>1$}.
    \end{cases}
    \end{equation}

    \subsection{Local parametrix near $1$}\label{section: local}

For $0<t<t_0$ with $t_0$ fixed but sufficiently small, we will now construct a
parametrix $P$ satisfying the same jump conditions as $S$ in a neighborhood $U$ of
$1$ of a sufficiently small fixed radius and a matching condition with $N$ on the
boundary $\partial U$.

Assume that $\Psi(\zeta)$ solves the RH problem of Section \ref{PVsection}, and define
    \begin{equation}\label{def Phi}
    \Phi(\lambda;x)=e^{\frac{x}{4}\sigma_3}x^{-\beta\sigma_3}\Psi(\frac{\lambda}{x}+\frac{1}{2};x)
G(\lambda;x)^{\frac{1}{2}\sigma_3}e^{\pm\frac{\pi i}{2}(\alpha-\beta)\sigma_3}, \qquad
\mbox{ for $\pm\Im\lambda>0$}
    \end{equation}
respectively,  with
    \begin{equation}\label{def G}
    G(\lambda;x)=(\lambda+\frac{x}{2})^{-(\alpha-\beta)}(\lambda-\frac{x}{2})^{\alpha+\beta}
    e^{\lambda}e^{-\pi i(\al-\bt)},\qquad x>0,
    \end{equation}
where $G$ is analytic in
$\mathbb C\setminus \left((-\infty,-\frac{x}{2}]\cup[\frac{x}{2},+\infty)\right)$.
We choose $-\pi<\arg(\lambda+\frac{x}{2})<\pi$ and $0<\arg(\lambda-\frac{x}{2})<2\pi$.
It is straightforward to check that $\Phi=\Phi(\lambda;x)$ solves the
following RH problem for $x>0$.
    \subsubsection*{RH problem for $\Phi$}
    \begin{itemize}
    \item[(a)] $\Phi:\mathbb C\setminus
\cup_{j=1}^4e^{\frac{\pi i (2j-1)}{4}}\mathbb R^+
 \to \mathbb C^{2\times 2}$ is analytic,
    with the rays $e^{\frac{\pi i (2j-1)}{4}}\mathbb R^+$ oriented as shown in
Figure \ref{figure: Phi}.
    \item[(b)] $\Phi$ has continuous boundary values on $\cup_{j=1}^4e^{\frac{\pi i
(2j-1)}{4}}\mathbb R^+\setminus\{0\}$, and they are related by the jump
conditions:
    \begin{align}
    &\label{RHP hatPhi:b1}
    \Phi_+(\lambda)=\Phi_-(\lambda)\begin{pmatrix}1&G(\lambda;x)^{-1}\\0&1\end{pmatrix},
&\mbox{ as $\lambda\in e^{\frac{\pi i}{4}}\mathbb R^+\cup e^{\frac{7\pi
i}{4}}\mathbb R^+,$}\\
    &\label{RHP
hatPhi:b2}\Phi_+(\lambda)=\Phi_-(\lambda)\begin{pmatrix}1&0\\-G(\lambda;x)&1\end{pmatrix},
&\mbox{ as $\lambda\in e^{\frac{3\pi i}{4}}\mathbb R^+\cup e^{\frac{5\pi
i}{4}}\mathbb R^+$.}\end{align}
    \item[(c)] $\Phi$ has the following behavior as $\lambda\to\infty$:
    \begin{equation}\label{RHP hatPhi: c}
    \Phi(\lambda)=I+\bigO(\lambda^{-1}).
    \end{equation}
\item[(d)] $\Phi$ is bounded near $0$.
\end{itemize}

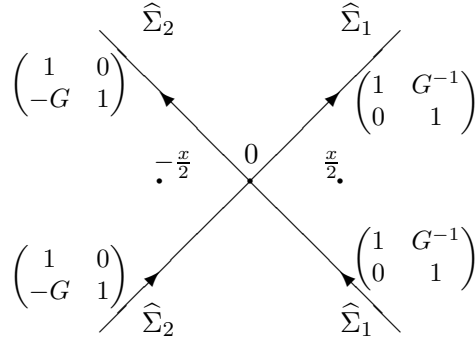
\begin{figure}[t]
\begin{center}
    \setlength{\unitlength}{0.8truemm}
    \begin{picture}(100,48.5)(0,2.5)

    \put(35,25){\thicklines\circle*{.8}}
    \put(34,27){\small $-\frac{x}{2}$}
    \put(62,27){\small $\frac{x}{2}$}
    \put(65,25){\thicklines\circle*{.8}}
    \put(50,25){\thicklines\circle*{.8}}
    \put(49,28){$0$}

    \put(50,25){\line(1,1){25}}
    \put(50,25){\line(1,-1){25}}
    \put(50,25){\line(-1,1){25}}
    \put(50,25){\line(-1,-1){25}}

    \put(65,40){\thicklines\vector(1,1){.0001}}
    \put(65,10){\thicklines\vector(-1,1){.0001}}
    \put(35,40){\thicklines\vector(-1,1){.0001}}
    \put(35,10){\thicklines\vector(1,1){.0001}}

    \put(65,50){$\widehat\Sigma_1$}
    \put(65,0){$\widehat\Sigma_1$}
    \put(32,50){$\widehat\Sigma_2$}
    \put(32,0){$\widehat\Sigma_2$}

    \put(67,37){\small $\begin{pmatrix}1&G^{-1}\\0&1\end{pmatrix}$}
    \put(10,40){\small $\begin{pmatrix}1&0\\-G&1\end{pmatrix}$}

    \put(10,8){\small $\begin{pmatrix}1&0\\-G&1\end{pmatrix}$}
    \put(67,11){\small $\begin{pmatrix}1&G^{-1}\\0&1\end{pmatrix}$}

    \end{picture}
    \caption{The jump contour and jump matrices for $\Phi$.}
    \label{figure: Phi}
\end{center}
\end{figure}

We will prove the following results.
\begin{proposition}\label{prop Phi}
\begin{itemize}
\item[{\rm (i)}] If $\Re\alpha>-\frac{1}{2}$, the RH problem for $\Phi$ is uniquely
solvable for all but possibly a finite number of positive
$x$-values $\{x_1, \ldots , x_k\}$, where $x_j=x_j(\alpha,\beta)$ and
$k=k(\alpha,\beta)$.
\item[{\rm (ii)}] If $\alpha>-\frac{1}{2}$ ($\Im\al=0$) and $\Re\beta=0$,
the RH problem for $\Phi$ is (uniquely) solvable for all $x>0$.
\item[{\rm (iii)}] If $\Re\alpha>-\frac{1}{2}$, the asymptotic condition (\ref{RHP
hatPhi: c}) for $\Phi$ is valid uniformly for $x\in (0,+\infty)$ provided that $x$
remains bounded away from
the set $\{x_1, \ldots , x_k\}$.
\end{itemize}
\end{proposition}
Statements (i) and (ii) follow immediately from Theorem \ref{theorem: Painleve}
(which will be proven in Section \ref{section: Painleve}).
The third statement will follow from our asymptotic analysis of the RH problem for
$\Phi$ in Section \ref{section: Painleve}.

    We will now transform the jump matrices for $\Phi$ into the jump matrices for $S$
near $1$. Note first that the off-diagonal entries of the jump matrices for $\Phi$
have branch points at $\pm\frac{x}{2}$, and the ones for $S$ at $e^{\pm t}$.
    Let us therefore define a conformal mapping $\lambda(z)$ in a neighborhood of
$1$ which maps $e^{-t}$ to $-\frac{x}{2}$, $e^{t}$ to $\frac{x}{2}$, and $1$ to
$0$:
\begin{equation}\label{def lambda}
    \lambda(z)=\frac{x}{2t}\ln(z), \qquad \mbox{ $z\in U$.}
    \end{equation}
Here we take the branch of the logarithm such that
$\ln z>0$ for $z>1$, and the branch cut is along the negative real axis.
We will furthermore need that $e^{\lambda(z)}=z^n$, and therefore set
\be
x=2nt.
\ee

Let us choose the contours $\Sigma_1$ and $\Sigma_2$ near $1$ in such a way that
$\lambda$ maps $\Sigma_1\cup\Sigma_2$ onto the jump contour
$\cup_{j=1}^4e^{\frac{\pi i (2j-1)}{4}}\mathbb R^+$  for $\Phi$. We look for the
parametrix $P$ in the form
    \begin{equation}
    \label{def P}
    P(z)=E(z)\Phi(\lambda(z);2nt)W(z),
    \end{equation}
    where $E$ is an analytic function in $U$, and $W$ is given by
    \begin{equation}\label{def W}
W(z)=\begin{cases}
-G(\lambda(z))^{-\frac{1}{2}\sigma_3}
z^{\frac{n}{2}\sigma_3}f(z)^{-\frac{1}{2}\sigma_3}\sigma_3, &\mbox{ for $|z|<1$,}\\
G(\lambda(z))^{-\frac{1}{2}\sigma_3}z^{\frac{n}{2}\sigma_3}f(z)^{\frac{1}{2}\sigma_3}\sigma_1,
&\mbox{ for $|z|>1$,}
\end{cases}
\end{equation}
with $\sigma_1=\begin{pmatrix}0&1\\1&0\end{pmatrix}$.
Note that the branch points of $G$ cancel the ones for $f$ in $U$,
and $W$ is analytic in $U\setminus C$.

If $E$ is analytic in $U$, it is easy to check
using (\ref{RHP hatPhi:b1})--(\ref{RHP hatPhi:b2}) that $P(z)$ satisfies the same
jump conditions
as the matrix $S$ with the jump matrices given in (\ref{J1})--(\ref{J2}).
Since we evaluate $\Phi(\lambda;x)$ at $x=2nt$, we need to impose the condition that
$2nt$ does not belong to the set $\{x_1, \ldots , x_k\}$ of values at which the RH
problem for $\Phi$ is not
solvable.

To fix $E(z)$, let us consider the behavior of $P$ on $\partial U$.
From (\ref{def lambda}) one observes that there exists $c>0$ such
that for any $0<t<t_0$
\begin{equation}
|\lambda(z)|> cn, \qquad \mbox{ $z\in\partial U$}.
\end{equation}
As $n\to\infty$ and if $2nt$ stays bounded away from the set $\{x_1, \ldots ,
x_k\}$, we can
thus (by Proposition \ref{prop Phi}) use the asymptotic behavior (\ref{RHP hatPhi:
c}) for $\Phi$ to
conclude that
\begin{equation}\label{matching1}
P(z)=E(z)\left(I+\bigO(n^{-1})\right)W(z),\qquad\mbox{ as $n\to\infty$,}
\end{equation}
uniformly for $0<t<t_0$ and $z\in\partial U$.
If $t_0$ is sufficiently small, we can assume that $e^{\pm t}$ lie inside $U$ and at
a distance
bounded from below away from $\partial U$. Then we obtain from (\ref{def W}) and
(\ref{def G}) that
(here and in (\ref{Eas}) below, $\bigO(1)$ is a scalar matrix element)
\be\label{asW}
W(z)=
n^{-\bt\si_3}
\begin{cases}
\begin{pmatrix}
\bigO(1) & 0 \cr 0 & \bigO(1)
\end{pmatrix},& |z|<1\cr
\begin{pmatrix}
0 & \bigO(1) \cr \bigO(1) & 0
\end{pmatrix},& |z|>1
\end{cases}
\ee
as $n\to\infty$ uniformly for $0<t<t_0$ and uniformly for $z\in\partial U\setminus C$.

Now set
    \begin{equation}\label{def E}
    E(z)=N(z)W(z)^{-1}.
    \end{equation}
One verifies directly, using the jumps for $N$ and $W$ across $C$, that $E$ is analytic
in a full neighborhood $\overline U$ of $1$.
Furthermore, by (\ref{asW}), (\ref{def N}),
\be\label{Eas}
E(z)=\begin{pmatrix}
0 & \bigO(1) \cr \bigO(1) & 0
\end{pmatrix} n^{\bt\si_3}
\ee
as $n\to\infty$ uniformly for $0<t<t_0$ and $z\in\partial U$.

Using this result and (\ref{matching1}), we obtain the following matching
condition on $z\in\partial U$:
    \begin{equation}\label{matching3}
    P(z)N(z)^{-1}=E(z)\left(I+\bigO(n^{-1})\right)E(z)^{-1}=
I+n^{-\bt\si_3}\bigO(n^{-1})n^{\bt\si_3}.
    \end{equation}
as $n\to\infty$ uniformly for $0<t<t_0$ and $z\in\partial U$.
Note once again that the matching holds true if $2nt$ remains bounded away from the set
$\{x_1, \ldots , x_k\}$.

Since $P$ has the same jumps as $S$ inside $U$
and $S(z)P(z)^{-1}=\bigO(\ln(z-1))$ as $z\to1$ for $t>0$,
it follows that the singularity is removable and
$S(z)P(z)^{-1}$ is analytic  in $U$.
For later use, we note that
\begin{equation}\label{E}
E(e^{t})=n^{-\beta\si_3}g^{\sigma_3}\si_1, \qquad
E(e^{-t})=n^{-\beta\si_3}h^{\sigma_3}\si_1,
\end{equation}
with $g$ and $h$ given by
\begin{align}
&\label{g1}g=\left(\frac{\sinh t}{t}\right)^{-\frac{\alpha-\beta}{2}}
e^{-\pi i\beta}e^{t\alpha}
\exp\left(-\frac{1}{2}\sum_{k=-\infty}^{-1}V_ke^{tk}+\frac{1}{2}\sum_{k=0}^{\infty}V_ke^{tk}\right),\\
&\label{h1}h=\left(\frac{\sinh t}{t}\right)^{\frac{\alpha+\beta}{2}}
e^{-\pi i\beta}
\exp\left(-\frac{1}{2}\sum_{k=-\infty}^{-1}V_ke^{-tk}+\frac{1}{2}\sum_{k=0}^{\infty}V_ke^{-tk}
\right).
\end{align}

    \subsection{Final RH problem}

    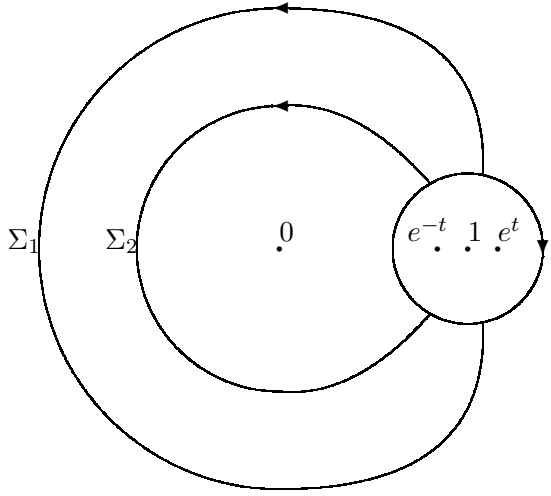
\begin{figure}[t]
    \begin{center}
    \setlength{\unitlength}{1truemm}
    \begin{picture}(100,55)(-5,10)
        \cCircle(75,40){10}[f]
        \put(50,40){\thicklines\circle*{.8}}
        \put(50,41){$0$}
        \put(75,40){\thicklines\circle*{.8}}
        \put(75,41){$1$}
        \put(71,40){\thicklines\circle*{.8}}
        \put(79,41){$e^t$}
        \put(79,40){\thicklines\circle*{.8}}
        \put(67,41){$e^{-t}$}
        \put(49,72){\thicklines\vector(-1,0){.0001}}
        \put(49,59){\thicklines\vector(-1,0){.0001}}

        \put(13.8,40){$\Sigma_1$}
        \put(26.8,40){$\Sigma_2$}


        \qbezier(77,50)(78,72)(50,72)
        \cCircle(50,40){32}[l]

        \qbezier(77,30)(78,8)(50,8)

\put(85,39){\thicklines\vector(0,-1){.0001}}

        \qbezier(70,48.7)(60,60)(50,59)
        \cCircle(50,40){19}[l]

        \qbezier(70,31.3)(60,20)(50,21)

    \end{picture}
    \caption{The contour $\Sigma_R$.}
\label{figure: R}
\end{center}
\end{figure}

    Define
    \begin{equation}\label{def R}
    R(z)=\begin{cases}
    n^{\bt\si_3}S(z)N(z)^{-1}n^{-\bt\si_3}, &\mbox{ for $z\in \mathbb C\setminus U$,}\\
    n^{\bt\si_3}S(z)P(z)^{-1}n^{-\bt\si_3}, &\mbox{ for $z\in U$.}
    \end{cases}
    \end{equation}

    Using the RH properties of $S$, $N$, and $P$, we obtain the following.
    \subsubsection*{RH problem for $R$}
    \begin{itemize}
    \item[(a)] $R$ is analytic in $\mathbb C\setminus \Sigma_R$, where
$\Sigma_R$ is the union of $\partial U$ and the parts of $\Si_1$, $\Si_2$
 lying outside $U$ (see Figure \ref{figure: R}).
    \item[(b)] $R_+(z)=R_-(z)J_R(z)$ for $z\in\Sigma_R$, where
    \begin{align}
    &\label{JR}
    J_R(z)=n^{\bt\si_3}P(z)N(z)^{-1}n^{-\bt\si_3}, &\mbox{ for $z\in\partial U$,}\\
    &J_R(z)=n^{\bt\si_3}N(z)J_k(z) N(z)^{-1}n^{-\bt\si_3}, &\mbox{ for
$z\in\Sigma_k$ outside $U$},
    \end{align}
and $J_k(z)$, $k=1,2$ are the jump matrices (\ref{J1}), (\ref{J2}) of $S$.
    \item[(c)] As $z\to\infty$,  $R(z)=I+\bigO(z^{-1})$.
    \end{itemize}

Using (\ref{J1}), (\ref{J2}), and (\ref{matching3}), we observe a crucial fact:
\begin{align}
    &\label{as JR}J_R(z)=I+\bigO(n^{-1}), &\mbox{ for $z\in\partial U$,}\\
    &J_R(z)=I+\bigO(e^{-cn}), &\mbox{ for $z\in\Sigma_R\setminus\partial U$, $c>0$}
    \end{align}
as $n\to\infty$ uniformly in $z$ and
uniformly for $0<t<t_0$ as long as $2nt$ remains bounded away from the set $\{x_1,
\ldots , x_k\}$.
Thus the jump matrix $J_R$ tends to the identity matrix as $n\to\infty$. The RH
problem for $R$ is
therefore a so-called small-norm RH problem, and by a standard analysis
(see, e.g., \cite{DKMVZ1}) we obtain the following statement.

\begin{proposition}\label{prop R}
Let $0<t<t_0$. Then
    \begin{itemize}
    \item[{\rm (i)}] The RH problem for $R$ is solvable for $n$ sufficiently large
as long as $2nt$ remains bounded away from the set $\{x_1, \ldots , x_k\}$.
    \item[{\rm (ii)}] If $n\to\infty$,
    \begin{equation}\label{expansion R}
    R(z)=I+\bigO(n^{-1}), \qquad \mbox{ uniformly for $z\in\mathbb C\setminus
\Sigma_R$}
    \end{equation}
and for $0<t<t_0$ such that $2nt$ remains bounded away from the set $\{x_1, \ldots ,
x_k\}$.
    \end{itemize}
    \end{proposition}

\subsection{Asymptotics for $\frac{d}{dt}\ln D_n$}

Reversing the transformations $S\mapsto R$, $T\mapsto S$, and $Y\mapsto T$
by  (\ref{def T}), (\ref{def S}), and (\ref{def R}),
we obtain the asymptotics for $Y(z)$ under the conditions of Proposition \ref{prop R}.
In particular, we obtain
\begin{align}\label{Yas1}
&Y(z)=n^{-\bt\si_3}(I+\bigO(n^{-1}))n^{\bt\si_3}P(z)z^{n\si_3},&\mbox{for $z$ near
$e^t$,} \\
&\label{Yas2}Y(z)=n^{-\bt\si_3}(I+\bigO(n^{-1}))n^{\bt\si_3}P(z),&\mbox{for $z$ near
$e^{-t}$}.
\end{align}
as $n\to\infty$, uniformly for $0<t<t_0$ if $2nt$ remains bounded away from $\{x_1,
\ldots , x_k\}$.
Using the definitions (\ref{def N}), (\ref{def P}), and (\ref{E}), we find for $P$
in the above formulas:
\begin{align}&\label{Pform}
P(z)=E(z)\Phi(z)W(z)=D(z)^{\si_3}W(z)^{-1}\Phi(z)W(z),\qquad\qquad\mbox{for $z$ near
$e^t$,} \\
&P(z)=E(z)\Phi(z)W(z)=D(z)^{\si_3}
\begin{pmatrix}
0&1\cr -1&0
\end{pmatrix}
W(z)^{-1}\Phi(z)W(z),\quad\mbox{for $z$ near $e^{-t}$}.
\end{align}

We will now substitute the asymptotics we obtained for $Y$ into the differential
identity (\ref{differentialidentity}) for $\ln D_n(t)$.
First, consider the case of $z$ close to $e^t$.
By (\ref{Yas1}), we obtain
\be
Y^{-1}Y'_z=
\frac{n\si_3}{z}+z^{-n\si_3}P^{-1}P'_z z^{n\si_3}+
z^{-n\si_3}P^{-1}(z)n^{-\bt\si_3}(I+\bigO(n^{-1}))^{-1}\bigO(n^{-1})'_z
n^{\bt\si_3}P(z)z^{n\si_3}.
\ee
Using (\ref{Pform}) and (\ref{def W}), we further obtain
\be
P^{-1}P'_z=-\si_3{A'_z\over A}+
W^{-1}\Phi^{-1}\Phi'_z W-W^{-1}\Phi^{-1}\si_3\Phi W \left( {A'_z\over A}+{D'_z\over D}
\right),
\ee
where we defined $A(z)$ by the formula
\[
W(z)=A(z)^{\si_3}\si_1.
\]
Expressions (\ref{def W}) and (\ref{def D2}) give
\begin{align}
{A'_z\over A}(e^t)&={\al+\bt\over 4}e^{-t}+{\al-\bt\over 4}e^{-t}
\left({1\over t}+{e^{-t}\over\sinh t}\right)
+{1\over 2}V'_z(e^t);\\
{D'_z\over D}(e^t)&=-{\al-\bt\over 2\sinh t}e^{-2t}-\sum_{k=-\infty}^{-1}kV_k
e^{(k-1)t}.
\end{align}
Therefore, we finally have for the 22 matrix element of $P^{-1}P'_z$ at the point
$e^t$:
\be\label{PPe}\eqalign{
e^t(P^{-1}P'_z)_{22}(e^t)=
{\al+\bt\over 4}+{\al-\bt\over 4}
\left({1\over t}+{e^{-t}\over\sinh t}\right)
+{1\over 2}e^t V'_z(e^t)
+e^t(\Phi^{-1}\Phi'_z)_{11}(e^t)-\\
\left({\al+\bt\over 4}+{\al-\bt\over 4}
\left({1\over t}-{e^{-t}\over\sinh t}\right)
+{1\over 2}e^t V'_z(e^t)-\sum_{k=-\infty}^{-1}kV_k e^{kt}
\right)(\Phi^{-1}\si_3\Phi)_{11}(e^t).}
\ee
Now using the definition of $W(z)$ it is easy to conclude that
\[
n^{\bt\si_3}P(e^t)=\hat\Phi(t)n^{\bt\si_3},
\]
where $\hat\Phi(t)$ is bounded in $n$ as long as $\Phi(\frac{x}{2})$ is bounded.
Thus, we obtained the asymptotic expression
\be\label{1Y}
e^t(Y^{-1}Y'_z)_{22}(e^t)=
-n+e^t(P^{-1}P'_z)_{22}(e^t)+\left(\hat\Phi^{-1}(t)\bigO(1/n)\hat\Phi(t)\right)_{22}
\ee
uniformly for $0<t<t_0$
as long as $2nt$
remains bounded away from the set $\{x_1, \ldots , x_k\}$,
with the second term on the r.h.s. given by (\ref{PPe}).

Similar calculations at $e^{-t}$ give
\be\label{2Y}
e^{-t}(Y^{-1}Y'_z)_{22}(e^{-t})=
e^{-t}(P^{-1}P'_z)_{22}(e^{-t})+\left(\hat\Phi^{-1}(t)\bigO(1/n)\hat\Phi(t)\right)_{22},
\ee
with
\be\label{PPem}\eqalign{
e^{-t}(P^{-1}P'_z)_{22}(e^{-t})=
-{\al-\bt\over 4}-{\al+\bt\over 4}
\left({1\over t}+{e^{-t}\over\sinh t}\right)
+{1\over 2}e^{-t} V'_z(e^{-t})
+e^{-t}(\Phi^{-1}\Phi'_z)_{22}(e^{-t})-\\
\left({\al-\bt\over 4}+{\al+\bt\over 4}
\left({1\over t}-{e^{-t}\over\sinh t}\right)
-{1\over 2}e^{-t} V'_z(e^{-t})+\sum_{k=1}^\infty kV_k e^{-kt}
\right)(\Phi^{-1}\si_3\Phi)_{22}(e^{-t}).}
\ee

Collecting (\ref{1Y}) and (\ref{2Y}) together, substituting into
(\ref{differentialidentity}), and noting that
\[
\Phi'_z={n\over z} \Phi'_\lb
\]
gives

\begin{proposition}
Let
\be\label{defw}
w(x)=-{\al+\bt\over 2}(\Phi^{-1}\Phi'_\lb)_{11}(x/2)+
{\al-\bt\over 2}(\Phi^{-1}\Phi'_\lb)_{22}(-x/2).
\ee
Then
\be\label{di-interim}\eqalign{
\frac{d}{dt}\ln D_n=
(\al+\bt)n-{\al^2+\bt^2\over 2}-{\al^2-\bt^2\over 2}
\left({1\over t}+{e^{-t}\over\sinh t}\right)\\
-{\al+\bt\over 2}e^t V'_z(e^t)+{\al-\bt\over 2}e^{-t} V'_z(e^{-t})+
2nw(x)\\
+{\al+\bt\over 2}
\left\{ {\al+\bt\over 2}+{\al-\bt\over 2}
\left({1\over t}-{e^{-t}\over\sinh t}\right)
+\sum_{k=1}^\infty k(V_k e^{kt}+V_{-k} e^{-kt})
\right\}(\Phi^{-1}\si_3\Phi)_{11}(e^t)\\
-{\al-\bt\over 2}
\left\{ {\al-\bt\over 2}+{\al+\bt\over 2}
\left({1\over t}-{e^{-t}\over\sinh t}\right)
+\sum_{k=1}^\infty k(V_k e^{-kt}+V_{-k} e^{kt})
\right\}(\Phi^{-1}\si_3\Phi)_{22}(e^{-t})\\+O(1/n)\wt\Phi(x),}
\ee
where the error term is uniform for $0<t<t_0$
as long as $2nt$
remains bounded away from the set $\{x_1, \ldots , x_k\}$, and $\wt\Phi(x)$
depends on $\Phi(x/2)$, $\Phi(-x/2)$, $\al$, $\bt$ only and is bounded when
these parameters are in a compact set.
\end{proposition}

In the next section, we will analyze the Painlev\'e functions $\Phi(\lb)$. Namely,
we will obtain their behavior at $x=0$ and $\infty$ which will be used in the last
section to prove
Theorems \ref{theorem: Toeplitz} and \ref{theorem: Toeplitz2}.

\section{Model RH problem near $z=1$ and the fifth Painlev\'e
equation}\label{section: Painleve}

Recall the RH problem for $\Psi$ stated in the introduction on the contour $\Gamma$
given in
Figure \ref{figure: Gamma} for $x>0$ and $\Re\alpha>-\frac{1}{2}$.
In this section, we will
analyze the $\Psi$-RH problem asymptotically for $x$ near zero and infinity and
prove Theorem \ref{theorem: Painleve} and Proposition \ref{prop Phi}.
Moreover, we will give an explicit formula for $w$ (defined by (\ref{defw}))
in terms of the Painlev\'e V function $v$.

For simplicity, we will omit the dependence of $\Psi$ on $\alpha$, $\beta$, and $x$
in our notation
when convenient.
The behavior of $\Psi$ near $0$ and $1$ can also be characterized in a different way
from conditions (\ref{RHP Psi: d0}), (\ref{RHP Psi: d1}).
The following statement holds.

    \begin{proposition}\label{prop Psi}
    Let $\Psi=\Psi(\zeta;x,\alpha,\beta)$ satisfy the
conditions (a), (b), (d0), and (d1) of the RH problem for $\Psi$. Set
    \begin{equation}
    \label{def
Psi0}\Psi_0(\zeta):=\Psi(\zeta)(\zeta-1)^{\frac{\alpha+\beta}{2}\sigma_3}\zeta^{-\frac{\alpha-\beta}{2}\sigma_3}.
    \end{equation}
Then the function $\Psi_0$ is analytic near $0$ and near $1$. The branch cuts for
$\zeta^{-\frac{\alpha-\beta}{2}\sigma_3}$ and
$(\zeta-1)^{\frac{\alpha+\beta}{2}\sigma_3}$ are chosen here along $[0,+\infty)$ and
$[1,+\infty)$, respectively.
    \end{proposition}
    \begin{proof}
The fact that $\Psi_0$ is analytic near $0$ and $1$ can be verified using the jump
conditions
for $\Psi$ for $(\zeta-1)^{\frac{\alpha+\beta}{2}\sigma_3}$ and
$\zeta^{-\frac{\alpha-\beta}{2}\sigma_3}$.
The isolated singularities at $0$ and $1$ are removable because of the conditions
(d0) and (d1).
    \end{proof}

Recall the function $\Phi$ defined in terms of $\Psi$ in (\ref{def Phi})-(\ref{def G}).
It satisfies the RH problem given in Section \ref{section: local}.

We will now perform an asymptotic analysis of the RH problem for $\Phi$ as
$x\to +\infty$ and as $x\searrow 0$.

    \subsection{Asymptotics for $\Phi$ as $x\to +\infty$}

    Consider Figure \ref{Figure: +infty} and
    define
    \begin{equation}
    \wt\Phi(\zeta;x)=\Phi(x\zeta;x),
    \end{equation}
for $\zeta$ outside of
the two triangular regions $A$ and $B$. In these regions set
\begin{align*}
&\wt\Phi(\zeta;x)=\Phi(x\zeta;x)\begin{pmatrix}1&G(x\zeta;x)^{-1}\\0&1\end{pmatrix},
&\mbox{in region B,}\\
&\wt\Phi(\zeta;x)=\Phi(x\zeta;x)\begin{pmatrix}1&0\\G(x\zeta;x)&1\end{pmatrix},
&\mbox{in region A.}
\end{align*}
Now $\wt\Phi$ is defined in such a way that it has its jumps only on the solid lines
in Figure \ref{Figure: +infty}. We have

\begin{figure}[t]
\begin{center}
    \setlength{\unitlength}{0.8truemm}
    \begin{picture}(100,48.5)(0,2.5)

    \put(20,25){\thicklines\circle*{.8}}
    \put(50,25){\thicklines\circle*{.8}}
    \put(19,27){\small $-\frac{1}{2}$}
    \put(49,28){\small $0$}
    \put(78,27){\small $\frac{1}{2}$}
    \put(80,25){\thicklines\circle*{.8}}
    \put(65,40){\line(1,1){15}}
    \put(35,40){\line(-1,1){15}}
    \put(65,10){\line(1,-1){15}}
    \put(35,10){\line(-1,-1){15}}
    \put(65,10){\line(0,1){30}}
    \put(35,10){\line(0,1){30}}
    \put(40,25){A}
    \put(57,25){B}
\multiput(50,25)(1,1){15}{\thicklines\circle*{.1}}
\multiput(50,25)(1,-1){15}{\thicklines\circle*{.1}}
\multiput(50,25)(-1,1){15}{\thicklines\circle*{.1}}
\multiput(50,25)(-1,-1){15}{\thicklines\circle*{.1}}
    \put(26,52){$\widehat\Gamma_2$}
    \put(71,52){$\widehat\Gamma_1$}
    \put(35,26){\thicklines\vector(0,1){.0001}}
    \put(65,26){\thicklines\vector(0,1){.0001}}

    \end{picture}
    \caption{The jump contour for $\wt\Phi$.}
    \label{Figure: +infty}
\end{center}
\end{figure}
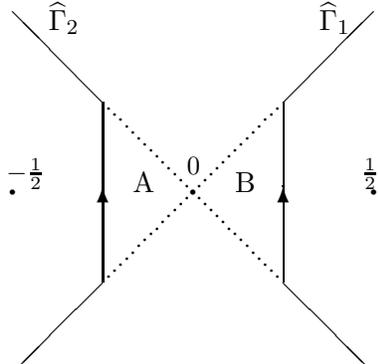

    \subsubsection*{RH problem for $\wt\Phi$}
    \begin{itemize}
    \item[(a)] $\wt\Phi:\mathbb C\setminus (\widehat\Gamma_1\cup\widehat\Gamma_2)
\to \mathbb C^{2\times 2}$ is analytic.
    \item[(b)] $\wt\Phi$ has continuous boundary values on
$\widehat\Gamma_1\cup\widehat\Gamma_2$ related by the conditions:
    \begin{align}
    &\label{RHP
hatPsi:b1}\wt\Phi_+(\zeta)=\wt\Phi_-(\zeta)\begin{pmatrix}1&G(x\zeta;x)^{-1}\\0&1\end{pmatrix},
    &\mbox{ for $\zeta\in\widehat\Gamma_1$,}\\
    &\label{RHP
hatPsi:b2}\wt\Phi_+(\zeta)=\wt\Phi_-(\zeta)\begin{pmatrix}1&0\\-G(x\zeta;x)&1\end{pmatrix},
    &\mbox{ as $\zeta\in\widehat\Gamma_2$.}
    \end{align}
    \item[(c)] $\wt\Phi$ has the following behavior as $\zeta\to\infty$,
    \begin{equation}\label{RHP wtPsi: c}
    \wt\Phi(\zeta)=I+\bigO(\zeta^{-1}).
    \end{equation}
    \end{itemize}

The jump matrices for $\wt\Phi$ are exponentially close to $I$ as $x\to +\infty$
because of the exponential factor in the definition of $G$ (see (\ref{def G})).
Indeed, let us  denote $G_{\wt\Phi}(\zeta)$ the jump matrix for  the function $\wt\Phi(z)$,
i.e.
$$
G_{\wt\Phi}(\zeta) = \begin{pmatrix}1&G(x\zeta;x)^{-1}\\0&1\end{pmatrix}, \quad \zeta \in \widehat\Gamma_1,
$$
and
$$
G_{\wt\Phi}(\zeta) = \begin{pmatrix}1&0\\-G(x\zeta;x)&1\end{pmatrix}, \quad \zeta \in \widehat\Gamma_2.
$$
Then, the following estimates hold:
\begin{equation}\label{100}
||I- G_{\wt\Phi} ||_{L_2(\widehat\Gamma)},\,\,||I- G_{\wt\Phi} ||_{L_{\infty}(\widehat\Gamma)}
< Ce^{-\frac{x}{2}(1-\epsilon)}, \quad 0 < \epsilon < 1,\quad C>0,
\end{equation}
where
$$
\widehat\Gamma = \widehat\Gamma_1 \cup \widehat\Gamma_2,
$$
and we assume
that the vertical parts of the contours $\widehat\Gamma_1$ and $\widehat\Gamma_2$ are
given by the equations $\Re \zeta = 1/2 -\epsilon/2$ and $\Re \zeta = -1/2 +\epsilon/2$,
respectively. Estimates (\ref{100}) imply that the RH
problem for $\wt\Phi$ is a small-norm RH problem for large $x$ and is
therefore solvable in a standard way (see e.g., \cite{DKMVZ1}) for $x$ sufficiently
large. Moreover, the solution $\wt\Phi$ admits the integral representation,
\begin{equation}\label{101}
\wt\Phi(\zeta) = I + \frac{1}{2\pi i}\int_{\widehat\Gamma}\rho(\zeta')
\Bigl(G_{\wt\Phi}(\zeta') - I\Bigr)\frac{d\zeta'}{\zeta'-\zeta},
\end{equation}
with the function $\rho(\zeta)$ which is $L_2$ - close to the identity, namely,
\begin{equation}\label{102}
||I- \rho ||_{L_2(\widehat\Gamma)} < Ce^{-\frac{x}{2}(1-\epsilon)}.
\end{equation}

The large $x$  solvability of the $\wt\Phi$--problem means that
 the RH problems for $\Phi$ and $\Psi$ are solvable for $x$ sufficiently
large as well.
In addition, we have that
    \begin{equation}\label{as wt Phi}
    \wt\Phi(\zeta;x)=I+\bigO(\frac{1}{\zeta}e^{-\frac{x}{2}(1-\epsilon)}), \qquad\mbox{ as $x\to
+\infty$}.
    \end{equation}
This estimate holds uniformly for $\zeta$ off the jump contour and this implies,
in particular, that the following asymptotics hold as $x\to +\infty$:
    \begin{align}
    &\label{as+0}\wt\Phi(\pm\frac{1}{2};x)=\Phi(\pm\frac{x}{2};x)=I+\bigO(e^{-\frac{x}{2}(1-\epsilon)}),\\
    &\label{as+1}\Phi'_\lb(\pm\frac{x}{2};x)=
\bigO(e^{-\frac{x}{2}(1-\epsilon)}).
    \end{align}
Furthermore, (\ref{as wt Phi}) implies that
 \begin{equation}\label{Phi infty}\Phi(\lambda; x)=I+\bigO(\lambda^{-1}),
\qquad\mbox{ as $\lambda\to\infty$, uniformly for $x>C$, $C>0$}.\end{equation}

The integral representation (\ref{101})
in conjunction with the estimate (\ref{102}) allows to evaluate the asymptotics
for the Painlev\'e function $v(x)$ defined in (\ref{def v intro}).
Let $\wt C_1$ be the first coefficient in the large $\zeta$
expansion of the function $\wt\Phi$,
\begin{equation}\label{103}
\wt\Phi(\zeta) =  I+ \frac{\wt C_1}{\zeta} + \bigO(\zeta^{-2}), \quad \zeta \to \infty.
\end{equation}
(Note that this expansion is uniform for $x > C$, $C > 0$.) Recalling the relation
between the functions $\wt\Phi$, $\Phi$, and $\Psi$, we obtain the following
expression of the first coefficient $C_1$ of the series (\ref{RHP Psi: c}) in
terms of the coefficient  $\wt C_1$,
\begin{equation}\label{104}
C_1 = \frac{\alpha + \beta}{2}\sigma_3
+e^{-\frac{x}{4}\sigma_3}x^{\beta\sigma_3}\wt C_1
x^{-\beta\sigma_3}e^{\frac{x}{4}\sigma_3}.
\end{equation}
Together with (\ref{def v intro}), this means that
\begin{equation}\label{105}
v(x) = - \wt C_{1,11} - x \wt C_{1,12} \wt C_{1,21}.
\end{equation}
On the other hand, from (\ref{101}) we obtain
\begin{equation}\label{106}
\wt C_1 = -\frac{1}{2\pi i}\int_{\widehat\Gamma}\rho(\zeta)
\Bigl(G_{\wt\Phi}(\zeta) - I\Bigr)\,d\zeta,
\end{equation}
which leads to the estimate:
\begin{eqnarray}
\wt C_1 &=& -\frac{1}{2\pi i}\int_{\widehat\Gamma}
\Bigl(G_{\wt\Phi}(\zeta) - I\Bigr)\,d\zeta -\frac{1}{2\pi
i}\int_{\widehat\Gamma}\Bigl(\rho(\zeta) - I\Bigr)
\Bigl(G_{\wt\Phi}(\zeta) - I\Bigr)\,d\zeta \nonumber\\
\label{vest1} &=& -\frac{1}{2\pi i}\int_{\widehat\Gamma}
\Bigl(G_{\wt\Phi}(\zeta) - I\Bigr)\,d\zeta  +
\bigO(e^{-x(1-\epsilon)}), \qquad 0<\epsilon <\frac{1}{2}.
\end{eqnarray}
The estimate (\ref{vest1}) implies the following asymptotic
representations for the entries of the matrix $\wt C_1$ as $x \to + \infty$:
\begin{eqnarray}
\wt C_{1,12}&=&-\frac{1}{2\pi i}\int_{\widehat\Gamma_1}
G(x\zeta;x)^{-1}\,d\zeta  +\bigO(e^{-x(1-\epsilon)})\nonumber\\
&=& -x^{-2\beta}e^{-\frac{x}{2}}\frac{e^{-2\pi
i\beta}}{\Gamma(\alpha+\beta)} \psi(1-\alpha -\beta, 2-2\beta; x) +
\bigO(e^{-x(1-\epsilon)})\nonumber\\
&=&-x^{-1+\alpha-\beta}e^{-\frac{x}{2}}\frac{e^{-2\pi
i\beta}}{\Gamma(\alpha+\beta)}\left(1 +
\bigO\left(\frac{1}{x}\right)\right),\label{vest2}
\end{eqnarray}
\begin{eqnarray}
\wt C_{1,21}&=&\frac{1}{2\pi i}\int_{\widehat\Gamma_2}
G(x\zeta;x)^{-1}\,d\zeta  +\bigO(e^{-x(1-\epsilon)})\nonumber\\
&=& x^{2\beta}e^{-\frac{x}{2}}\frac{e^{2\pi
i\beta}}{\Gamma(\alpha-\beta)} \psi(1-\alpha +\beta, 2+2\beta; x) +
\bigO(e^{-x(1-\epsilon)})\nonumber\\
&=&x^{-1+\alpha+\beta}e^{-\frac{x}{2}}\frac{e^{2\pi
i\beta}}{\Gamma(\alpha-\beta)}\left(1 +
\bigO\left(\frac{1}{x}\right)\right),\label{vest3}
\end{eqnarray}
and \begin{equation}\wt C_{1, 11}=-\wt C_{1,
22}=\bigO(e^{-\frac{x}{2}(1-\e)}),
\end{equation} where $\psi(a,c;x)$ denotes the confluent
hypergeometric function. The last estimate can be improved with the
help of the differential identity (\ref{qrt}) which will be proven
in Section \ref{section 43}. Indeed, this identity implies that
$$
\frac{d}{dx} \wt C_{1,11} = \wt C_{1,12}\wt C_{1,21},
$$
and hence
\begin{equation}
\wt C_{1,11} = x^{-2+2\alpha}e^{-x}
\frac{1}{\Gamma(\alpha-\beta)\Gamma(\alpha+\beta)}\left(1 +
\bigO\left(\frac{1}{x}\right)\right).\label{vest5}
\end{equation}
Substituting the estimates (\ref{vest2}), (\ref{vest3}), and (\ref{vest5}) into the formula
(\ref{105}), we arrive at the following asymptotic equation for the Painlev\'e function
$v(x)$:
\begin{equation}\label{vest6}
v(x) = x^{-1+2\alpha}e^{-x}
\frac{-1}{\Gamma(\alpha-\beta)\Gamma(\alpha+\beta)}
\left(1 + \bigO\left(\frac{1}{x}\right)\right).
\end{equation}

\subsection{Asymptotics for $\Psi$ and $\Phi$ as $x\searrow 0$}

    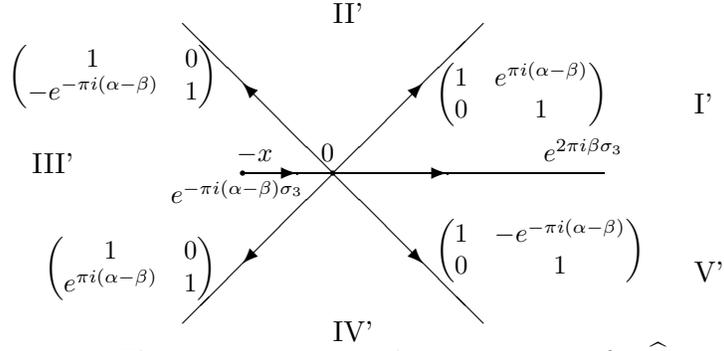
\begin{figure}[t]
\begin{center}
    \setlength{\unitlength}{0.8truemm}
    \begin{picture}(100,48.5)(0,2.5)

    \put(35,25){\thicklines\circle*{.8}}
    \put(34,27){\small $-x$}
    \put(48,27){\small $0$}
    \put(50,25){\thicklines\circle*{.8}}
    \put(35,25){\line(1,0){15}}
    \put(44,25){\thicklines\vector(1,0){.0001}}
    \put(50,25){\line(1,1){25}}
    \put(50,25){\line(1,-1){25}}
    \put(50,25){\line(-1,1){25}}
    \put(50,25){\line(-1,-1){25}}
    \put(50,25){\line(1,0){45}}
    \put(69,25){\thicklines\vector(1,0){.0001}}
    \put(65,40){\thicklines\vector(1,1){.0001}}
    \put(65,10){\thicklines\vector(1,-1){.0001}}
    \put(35,40){\thicklines\vector(-1,1){.0001}}
    \put(35,10){\thicklines\vector(-1,-1){.0001}}

    \put(67,37){\small $\begin{pmatrix}1&e^{\pi i(\alpha-\beta)}\\0&1\end{pmatrix}$}
    \put(-4,40){\small $\begin{pmatrix}1&0\\-e^{-\pi i(\alpha-\beta)}&1\end{pmatrix}$}
    \put(85,27){\small $e^{2\pi i\beta\sigma_3}$}
    \put(2,8){\small $\begin{pmatrix}1&0\\e^{\pi i(\alpha-\beta)}&1\end{pmatrix}$}
    \put(67,11){\small $\begin{pmatrix}1&-e^{-\pi i(\alpha-\beta)}\\0&1\end{pmatrix}$}
    \put(23,20){\small $e^{-\pi i(\alpha-\beta)\sigma_3}$}
    \put(110,35){I'}
    \put(50,50){II'}
    \put(0,25){III'}
    \put(110,7){V'}
    \put(50,-3){IV'}
    \end{picture}
    \caption{The jump contour and jump matrices for $\widehat\Psi$.}
    \label{Figure: 0}
\end{center}
\end{figure}
Write $\Psi_{\rm I}, \ldots, \Psi_{\rm V}$ for the analytic continuation of $\Psi$
from the
indicated in Figure \ref{figure: Gamma} sectors
I, $\ldots$, V  to $\mathbb C\setminus [0,+\infty)$, respectively, and consider the
function
\begin{equation}\label{def tildePsi}
\widehat\Psi(\lambda;x):=e^{\frac{x}{2}\sigma_3}x^{-\beta\sigma_3}\times
\begin{cases}
\Psi_{\rm I}(\frac{\lambda}{x}+1;x),&\mbox{ for $\lambda$ in region I',}\\
\Psi_{\rm II}(\frac{\lambda}{x}+1;x),&\mbox{ for $\lambda$ in region II',}\\
\Psi_{\rm III}(\frac{\lambda}{x}+1;x),&\mbox{ for $\lambda$ in region III',}\\
\Psi_{\rm IV}(\frac{\lambda}{x}+1;x),&\mbox{ for $\lambda$ in region IV',}\\
\Psi_{\rm V}(\frac{\lambda}{x}+1;x),&\mbox{ for $\lambda$ in region V',}
\end{cases}
\end{equation}
where the (modified) regions I', $\ldots$, V' are indicated in Figure \ref{Figure:
0}, so that $\widehat\Psi=\widehat\Psi(\lambda;x,\alpha,\beta)$ has jumps on a
contour which is partially shifted compared to the one for $\Psi(\lambda/x+1)$. In
particular, the intersection of the contour lines is now at $\lambda=0$
instead of $\lambda=-x/2$.

From the RH problem for $\Psi$, one easily derives the RH conditions
for $\widehat\Psi$.
\subsubsection*{RH problem for $\widehat\Psi$}
\begin{itemize}
    \item[(a)] $\widehat\Psi:\mathbb C\setminus\left(e^{\pm \frac{i\pi}{4}}\mathbb
R\cup\mathbb R^+\cup[-x,0]\right)\to \mathbb C^{2\times 2}$ is analytic.
    \item[(b)] $\widehat\Psi$ has continuous boundary values on $e^{\pm
\frac{i\pi}{4}}\mathbb R\cup\mathbb R^+\cup[-x,0]\setminus\{-x,0\}$, and they
are related as follows (with the orientation of the contour as in Figure
\ref{Figure: 0}),
    \begin{align}
    &\label{RHP
tildePsi:b1}\widehat\Psi_+(\lambda)=\widehat\Psi_-(\lambda)\begin{pmatrix}1&e^{\pi
i(\alpha-\beta)}\\0&1\end{pmatrix}, &\mbox{ as $\lambda\in
e^{\frac{i\pi}{4}}\mathbb R^+$,}\\
    &\widehat\Psi_+(\lambda)=\widehat\Psi_-(\lambda)\begin{pmatrix}1&0\\-e^{-\pi
i(\alpha-\beta)}&1\end{pmatrix}, &\mbox{ as $\lambda\in
e^{\frac{3i\pi}{4}}\mathbb R^+$,}\\
    &\widehat\Psi_+(\lambda)=\widehat\Psi_-(\lambda)\begin{pmatrix}1&0\\e^{\pi
i(\alpha-\beta)}&1\end{pmatrix}, &\mbox{ as $\lambda\in
e^{\frac{5i\pi}{4}}\mathbb R^+$,}\\
    &\widehat\Psi_+(\lambda)=\widehat\Psi_-(\lambda)\begin{pmatrix}1&-e^{-\pi
i(\alpha-\beta)}\\0&1\end{pmatrix}, &\mbox{ as $\lambda\in
e^{\frac{7i\pi}{4}}\mathbb R^+$,}\\
&\widehat\Psi_+(\lambda)=\widehat\Psi_-(\lambda)e^{2\pi i\beta\sigma_3}, &\mbox{ as
$\lambda\in\mathbb R^+$},\\
    &\label{RHP tildePsi:b6}\widehat\Psi_+(\lambda)=\widehat\Psi_-(\lambda)e^{-\pi
i(\alpha-\beta)\sigma_3}, &\mbox{ as $\lambda\in (-x,0)$}.
    \end{align}
    \item[(c)] As $\lambda\to\infty$,
    \begin{equation}\label{RHP tildePsi: c}
    \widehat\Psi(\lambda)=\left(I+\bigO(\lambda^{-1})\right)
\lambda^{-\beta\sigma_3}e^{-\frac{1}{2}\lambda\sigma_3}.
    \end{equation}
\item[(d0)]As $\lambda\to -x$,
    \begin{equation}\label{RHP tildePsi: d0}
    \widehat\Psi(\lambda)=\bigO\begin{pmatrix}|\lambda+x|^{\frac{\alpha-\beta}{2}}&|\lambda+x|^{-\frac{\alpha-\beta}{2}}\\
    |\lambda+x|^{\frac{\alpha-\beta}{2}}&|\lambda+x|^{-\frac{\alpha-\beta}{2}}\end{pmatrix}.
    \end{equation}
    \item[(d1)] As $\lambda\to 0$ in sectors I' and V',
    \begin{equation}\label{RHP tildePsi: d1}
    \widehat\Psi(\lambda)=\bigO\begin{pmatrix}|\lambda|^{-\frac{\alpha+\beta}{2}}&|\lambda|^{\frac{\alpha+\beta}{2}}\\
    |\lambda|^{-\frac{\alpha+\beta}{2}}&|\lambda|^{\frac{\alpha+\beta}{2}}\end{pmatrix}.
    \end{equation}
    As $\lambda\to 0$ in the other sectors, the behaviour of $\widehat\Psi$ is
obtained by
applying the jump conditions to (\ref{RHP tildePsi: d1}).
$\widehat\Psi(\lambda)=\bigO(\lambda^{-\frac{|\alpha+\beta|}{2}})$ always holds.
  \end{itemize}
For small values of $x$, we will now construct a global parametrix and a local
parametrix near $0$ for
$\widehat\Psi$ and match them on the boundary of an $\varepsilon$-neighborhood of
$\lambda=0$.
These constructions will lead to the uniform asymptotics for $\widehat\Psi$ and
$\Psi$ as $x\searrow 0$.

\subsubsection{Construction of the global parametrix}
Consider a fixed $\varepsilon$-neighborhood $U_\varepsilon$
of $\lambda=0$ containing, in particular, the $[-x,0]$ part
of the contour.
Outside of this neighborhood, we expect to model $\widehat\Psi$ by the
global parametrix $M=M(\lambda;\alpha,\beta)$ independent of $x$ and solving the
following RH problem.

\subsubsection*{RH problem for $M$}
    \begin{itemize}
    \item[(a)] $M:\mathbb C\setminus \left(e^{\pm\frac{\pi i}{4}}\mathbb R\cup
\mathbb R^+\right) \to \mathbb C^{2\times 2}$ is analytic,
    \item[(b)] $M$ has continuous boundary values on $e^{\pm\frac{\pi i}{4}}\mathbb
R\cup \mathbb R^+\setminus\left\{0\right\}$ related by the conditions
    \begin{align}
    &\label{RHP M0: b1}M_{+}(\lambda)=M_{-}(\lambda)\begin{pmatrix}1&e^{\pi
i(\alpha-\beta)}\\0&1\end{pmatrix}, &\mbox{ as $\lambda\in
e^{\frac{i\pi}{4}}\mathbb R^+$,}\\
    &M_{+}(\lambda)=M_{-}(\lambda)\begin{pmatrix}1&0\\-e^{-\pi
i(\alpha-\beta)}&1\end{pmatrix}, &\mbox{ as $\lambda\in
e^{\frac{3i\pi}{4}}\mathbb R^+$,}\\
    &M_{+}(\lambda)=M_{-}(\lambda)\begin{pmatrix}1&0\\e^{\pi
i(\alpha-\beta)}&1\end{pmatrix}, &\mbox{ as $\lambda\in
e^{\frac{5i\pi}{4}}\mathbb R^+$,}\\
    &M_{+}(\lambda)=M_{-}(\lambda)\begin{pmatrix}1&-e^{-\pi
i(\alpha-\beta)}\\0&1\end{pmatrix}, &\mbox{ as $\lambda\in
e^{\frac{7i\pi}{4}}\mathbb R^+$,}\\
    &\label{RHP M0: b5}M_{+}(\lambda)=M_{-}(\lambda)e^{2\pi i\beta\sigma_3}, &\mbox{
as $\lambda\in \mathbb R^+$}.
    \end{align}
    \item[(c)] We have
\begin{equation}\label{RHP Ph: c}
M(\lambda)=\left(I+\bigO(\lambda^{-1})\right)\lambda^{-\beta\sigma_3}e^{-\frac{1}{2}\lambda\sigma_3},\qquad
\mbox{ as $\lambda\to\infty$}.
\end{equation}
    \end{itemize}
We can solve this RH problem explicitly in terms of the confluent hypergeometric
function. Inspired by the constructions of \cite{IK1, DIK2}, we define
\be\label{PsiConfl1}
\eqalign{
H(\lambda):=\begin{pmatrix}
e^{-i\pi(2\bt+\al)}& 0 \cr 0 & e^{i\pi(\bt+2\al)}
\end{pmatrix}
e^{-\frac{i\pi}{2}\alpha\sigma_3}\left(\begin{matrix}
\lambda^{\alpha}\psi(\alpha+\beta,1+2\alpha,\lambda)e^{i\pi(2\beta+\alpha)} \cr
\lambda^{-\alpha}
\psi(1-\alpha+\beta,1-2\alpha,\lambda)e^{i\pi(\beta-3\alpha)}
{\Gamma(1+\alpha+\beta)\over\Gamma(\alpha-\beta)}
\end{matrix}
\right.\\
\left.
\begin{matrix}
\lambda^{\alpha}
\psi(1+\alpha-\beta,1+2\alpha,e^{-i\pi}\lambda)e^{i\pi(\beta+\alpha)}
{\Gamma(1+\alpha-\beta)\over\Gamma(\alpha+\beta)}
\cr
\lambda^{-\alpha}
\psi(-\alpha-\beta,1-2\alpha,e^{-i\pi}\lambda)e^{-i\pi\alpha}
\end{matrix}\right)e^{\frac{i\pi\alpha}{2}\sigma_3}e^{-\lambda\si_3/2},\qquad
\al\pm\bt\neq -1,-2,\dots,}
\ee
where $\psi(a,b,x)$ is the confluent hypergeometric function, and
$\Gamma(x)$ is Euler's $\Gamma$-function.
Furthermore let
\begin{align*}
&M(\lambda)=M_1(\lambda):=H(\lambda)\begin{pmatrix}1&-e^{\pi
i(\alpha-\beta)}\\0&1\end{pmatrix},&&\mbox{ for $0<\arg\lambda<\frac{\pi}{4}$},\\
&M(\lambda)=M_2(\lambda):=H(\lambda),&&\mbox{ for
$\frac{\pi}{4}<\arg\lambda<\frac{3\pi}{4}$},\\
&M(\lambda)=M_3(\lambda):=H(\lambda)\begin{pmatrix}1&0\\-e^{-\pi
i(\alpha-\beta)}&1\end{pmatrix},&&\mbox{ for
$\frac{3\pi}{4}<\arg\lambda<\frac{5\pi}{4}$},\\
&M(\lambda)=M_4(\lambda):=H(\lambda)\begin{pmatrix}1&0\\2i\sin\pi(\alpha-\beta)&1\end{pmatrix},&&\mbox{
for $\frac{5\pi}{4}<\arg\lambda<\frac{7\pi}{4}$},\\
&M(\lambda)=M_5(\lambda):=H(\lambda)\\
&\ \times \  \begin{pmatrix}1&-e^{-\pi
i(\alpha-\beta)}\\2i\sin\pi(\alpha-\beta)&-2ie^{-\pi
i(\alpha-\beta)}\sin\pi(\alpha-\beta)+1\end{pmatrix},&&\mbox{ for
$\frac{7\pi}{4}<\arg\lambda<2\pi$.}
\end{align*}
Using known properties of the confluent hypergeometric function, one verifies as in
\cite{DIK2} that $M$ satisfies the prescribed RH conditions.

In order to match $M$ later on with the local parametrix near zero that we construct
in the next section,
we will now need to rewrite $M$ in the form in which the structure of its
singularity at $\lambda=0$
becomes more apparent.

Recall the following properties of the confluent hypergeometric function (see, e.g.,
\cite{BE}):
\begin{align}
&\frac{\phi(a,c,z)}{\Gamma(c)}=
\frac{e^{i\pi a}}{\Gamma(c-a)}\psi(a,c,z)+
\frac{e^{-i\pi(c-a)}}{\Gamma(a)}\psi(c-a,c,e^{-i\pi}z)e^z,\label{ci1}\\
&\psi(a,c,z)=z^{1-c}\psi(a-c+1,2-c,z),\label{ci2}\\
&\psi(a,c,z)=
\frac{\Gamma(1-c)}{\Gamma(a-c+1)}\varphi(a,c,z)+\frac{\Gamma(c-1)}{\Gamma(a)}z^{1-c}
\varphi(a-c+1,2-c,z),\qquad c\notin\mathbb Z,\label{ci3}
\\
&\varphi(a,c,z)=e^z\varphi(c-a,c,-z).\label{ci4}
\end{align}
where
\begin{equation}\label{phidef}
\varphi(a,c;z) =
1+\sum_{n=1}^{\infty}\frac{a(a+1)\cdots(a+n-1)}{c(c+1)\cdots(c+n-1)}\frac{z^n}{n!},
\qquad c\neq 0,-1,-2,\dots
\end{equation}
is an entire function.
Let us focus on the region III'. Assume first that $2\al\neq 0,1,2,\dots$
(noninteger $2\al$ combined with our general condition $\Re\al>-1/2$).
Using the properties (\ref{ci1}) and (\ref{ci2}) for simplification of the first
column of $M_3$,
and (\ref{ci3}), (\ref{ci4}) for the second, we easily obtain
\be\label{M3}
M_3(\lb)=E(\lb)\lb^{\al\si_3}G_3,\qquad 2\al\neq 0,1,\dots,\quad \al\pm\bt\neq
-1,-2,\dots,
\ee
with the branch of $\lb^{\pm\al}$ chosen with $0<\arg\lb<2\pi$.
Here
\begin{multline}\label{E1}
E(\lb)=e^{-\lb/2}\left(\begin{matrix}
e^{-i\pi(\al+\bt)}\frac{\Gamma(1+\al-\bt)}{\Gamma(1+2\al)}\varphi(\al+\bt,1+2\al,\lb)
\cr
-e^{-i\pi(\al-\bt)}\frac{\Gamma(1+\al+\bt)}{\Gamma(1+2\al)}\varphi(1+\al+\bt,1+2\al,\lb)
\end{matrix}\right.\\
\left.\begin{matrix}
e^{i\pi(\al-\bt)}\frac{\Gamma(2\al)}{\Gamma(\al+\bt)}\varphi(-\al+\bt,1-2\al,\lb)\cr
e^{i\pi(\al+\bt)}\frac{\Gamma(2\al)}{\Gamma(\al-\bt)}\varphi(1-\al+\bt,1-2\al,\lb)
\end{matrix}\right)
\end{multline}
is entire, and $G_3$ is the constant matrix
\be\label{G3}
G_3=\begin{pmatrix}
1 & c_1 \cr
0 & 1
\end{pmatrix},\qquad c_1=-\frac{\sin\pi(\alpha+\beta)}{\sin 2\pi\alpha}.
\ee
Applying to (\ref{M3}) the jump conditions we readily obtain the general formulas
\be\label{M0}
 M_j(\lambda)=E(\lambda)\lambda^{\alpha\sigma_3}G_j,\qquad j=1, \ldots , 5,
\qquad  2\al\neq 0,1,\dots,\quad \al\pm\bt\neq -1,-2,\dots,
\ee
with appropriate constant matrices $G_j$. In particular, $G_3$ is given by
(\ref{G3}) and
\be\label{G1}
G_1=\begin{pmatrix}
1+c_1 e^{-i\pi(\al-\bt)} & -e^{i\pi(\al-\bt)} \cr
e^{-i\pi(\al-\bt)} & 0
\end{pmatrix}.
\ee

Consider now the case when $2\al$ is an integer. In this case, the calculations for the
first column of $M_3$ remain the same, whereas for the second column (note that
(\ref{ci3}) does
not hold now) we use the known logarithmic formulas for the $\psi$-function. We then
obtain
\begin{align}
&M_j(\lambda)=\wt E(\lambda)\lb^{\al\si_3}
\begin{pmatrix}1&\ga(\lb)\\0& 1\end{pmatrix}\wt G_j, &j=1, \ldots , 5,\\
&\ga(\lb)=\frac{(-1)^{2\al+1}}{\pi}\sin\pi(\alpha+\beta)\ln(\lb e^{-i\pi}),
&\mbox{ if }\quad 2\alpha=0,1,\dots,
\end{align}
where $\wt E(\lb)$ is analytic at zero, and $\wt G_j$ are constant matrices. In
particular,
$\wt G_3=I$.

\subsubsection{Construction of the local parametrix near $0$}

We construct a local parametrix $P$ in $U_\varepsilon$ in such a way that it has the
singularities and jumps of
$\widehat\Psi$ and matches with $M$ to the main order in $x$ on the boundary.
We have to consider the cases $2\al\notin\mathbb Z$ and $2\al\in\mathbb Z$ separately
because of the different behavior of $M$ at $\partial U_\varepsilon$.
As examination of the final formulas show (see (\ref{J0}), (\ref{J11}), (\ref{wtJ0}),
and (\ref{wtJ1}) below), our constructions in this section will be valid for all $\al$,
$\bt$ such that $\al\pm\bt\neq -1,-2,\dots$ This is exactly the restriction on $M$ in
the previous section. Some preliminary expressions, however, are valid under stronger
conditions $\al\pm\bt\notin \mathbb Z$  (cf. (\ref{F1})). We do not mention these conditions
as they disappear in the final formulas, namely: the singularities in $\al$, $\bt$
of $J$ and $\wt J$ defined below cancel with the zeros of $c_0$, $c_2$.

We first deal with the case when $2\alpha\neq 0,1,2,\dots$.
Since we have a problem with 2 singular points with power-law behaviour at $-x$ and
$0$ and
require a power-law behaviour at the boundary $\partial U_\varepsilon$ (see (\ref{M0}))
we expect from the general principles (e.g., \cite{FIKN}) a parametrix in terms of
the hypergeometric function.
Indeed such a parametrix was found by Jimbo \cite{J} in the generic case of
Painlev\'e V equation.
Instead of trying to specialize it to our situation, we provide a direct
construction below.
For $2\alpha\notin\mathbb Z$ , define $P=P(\lambda;x,\alpha,\beta)$ in
$U_\varepsilon$ by
the expressions
\begin{equation}
P(\lambda)=P_j(\lambda),\qquad
\mbox{ with $j=1$ for $\lambda$ in sector I', $j=2$ in sector II', and so on,}
\end{equation}
where
\begin{equation}\label{def Pb}
P_j(\lambda)=E(\lambda)\begin{pmatrix}1&
c_0 J(\lb;x,\al,\bt)\\0&1\end{pmatrix}(\lambda+x)^{\frac{\alpha-\beta}{2}\sigma_3}
\lambda^{\frac{\alpha+\beta}{2}\sigma_3}G_j
\end{equation}
with the argument of the roots between $0$ and $2\pi$.
Here $E$ is given by (\ref{E1}), $G_j$ are as in (\ref{M0}),
\be\label{F}
J(\lb;x,\al,\bt)=\frac{1}{\pi}
\frac{x^{1+2\al}}{-\lb}\frac{\Ga(1+\al+\bt)\Ga(1+\al-\bt)}{\Ga(2+2\al)}
F(1,1+\al+\bt,2+2\al,e^{i\pi}x/\lb),
\ee
\begin{equation}\label{C0}
c_0=-e^{2\pi i\alpha}\frac{\sin\pi(\alpha+\beta)\sin\pi(\alpha-\beta)}{\sin
2\pi\alpha},
\end{equation}
and $F(a,b,c,z)$ is the hypergeometric function of $z$ with parameters $a$, $b$,
$c$. The argument of $z$ is chosen between $0$ and $2\pi$.
For $c\neq 0,-1,-2,\dots$, this function is
represented by the standard series
\[
F(a,b,c,z)= 1+\sum_{n=1}^{\infty}\frac{a(a+1)\cdots(a+n-1)b(b+1)\cdots(b+n-1)}
{c(c+1)\cdots(c+n-1)}\frac{z^n}{n!},
\]
converging in the disk $|z|\le r<1$ of any radius $r<1$,
and is extended to the analytic function in the plane with a cut $[1,+\infty)$.
Therefore, the function $F(a,b,c,e^{i\pi}x/\lb)$, in particular the one in (\ref{F}),
is analytic in $\lb$-plane outside $[-x,0]$.
We will now find the jump of $J(\lb)$ on $[-x,0]$ and the structure of its
singularities
at $-x$ and $0$. First, using the transformation of the hypergeometric functions from
the one with the argument $z$ to those with the argument $1/z$, we can write:
\be\label{F1}
\eqalign{
F(1,1+\al+\bt,2+2\al,z)=
-\frac{\pi}{\sin\pi(\al+\bt)}\frac{\Ga(2+2\al)}{\Ga(1+\al+\bt)\Ga(1+\al-\bt)}\\
\times\left(e^{i\pi} z^{-1}\right)^{1+\al+\bt}F(1+\al+\bt,-\al+\bt,1+\al+\bt,1/z)\\
+\frac{1+2\al}{\al+\bt}e^{i\pi} z^{-1}F(1,-2\al,1-\al-\bt,1/z).}
\ee
Since the function
\[
F(1+\al+\bt,-\al+\bt,1+\al+\bt,1/z)=\left(1-{1\over z}\right)^{\al-\bt},
\]
and $F(1,-2\al,1-\al-\bt,1/z)$ are analytic outside $[0,1]$, we have on $(1,+\infty)$:
\[\eqalign{
F(1,1+\al+\bt,2+2\al,z)_+ -F(1,1+\al+\bt,2+2\al,z)_-\\
=\frac{2\pi i\Ga(2+2\al)}{\Ga(1+\al+\bt)\Ga(1+\al-\bt)}
|z|^{-1-\al-\bt}\left|1-{1\over z}\right|^{\al-\bt},\qquad z\in (1,+\infty).}
\]
Setting here $z=e^{i\pi}x/\lb$ and using (\ref{F}), we obtain
\be\label{Jjump}
J_+(\lb)=J_-(\lb)+2i |\lb|^{\al+\bt}|\lb+x|^{\al-\bt},\qquad \lb\in (-x,0).
\ee

By (\ref{M0}) we can write
\begin{equation}\label{def P2}
P(\lambda)=E(\lambda)\begin{pmatrix}1& c_0 J(\lb)
\\0&1\end{pmatrix}
(\lambda+x)^{\frac{\alpha-\beta}{2}\sigma_3}
\lambda^{-\frac{\alpha-\beta}{2}\sigma_3}E(\lambda)^{-1}M(\lambda).
\end{equation}
Since the product of the factors to the left of $M$ is analytic in $\mathbb
C\setminus [-x,0]$, the
expression (\ref{def P2})
implies directly that $P$ satisfies the jump conditions for $M$, see (\ref{RHP M0:
b1})--(\ref{RHP M0: b5}),
except on $(-x,0)$ where $M$ is analytic but $P$ is not.
For $\lambda\in (-x,0)$, it is convenient to use (\ref{def Pb}), (\ref{G3}), and
(\ref{Jjump}) to verify
that
\[
P_+(\lambda)=P_-(\lambda)e^{-\pi i(\alpha-\beta)\sigma_3}.
\]
This is only true if $c_0$ is given by (\ref{C0}).
We have thus constructed $P$ in such a way that it has exactly the jump conditions
for $\widehat\Psi$.

\medskip

From the fact that $P$ and $\widehat\Psi$ have the same jumps, it follows that
$\widehat\Psi P^{-1}$ is
analytic in $U_\varepsilon$ except possibly at the points $-x$ and $0$.
Let us investigate the behavior of $\widehat\Psi P^{-1}$ near these points in some
detail.
Recall that by Proposition \ref{prop Psi}
$\Psi(\zeta)\zeta^{-\frac{\alpha-\beta}{2}\sigma_3}(\zeta-1)^{\frac{\alpha+\beta}{2}\sigma_3}$
is analytic
at $0$ and $1$, which implies by (\ref{def tildePsi}) that
$\widehat\Psi_0(\lambda):=\widehat\Psi(\lambda)(\lambda+x)^{-\frac{\alpha-\beta}{2}\sigma_3}
\lambda^{\frac{\alpha+\beta}{2}\sigma_3}$ is analytic at $-x$ and $0$, where at $-x$ the
variable $\lb$ is in the region $III'$, while at $0$, the variable $\lb$ is
in the region $I'$ or $V'$ (for other regions at $0$, the appropriate jump conditions
should be applied to $\widehat\Psi(\lambda)$). Consider first $\lb$ close to $0$.
By (\ref{def Pb}), we have
\begin{multline}\label{psip}
\widehat\Psi(\lambda)P(\lambda)^{-1}=
\widehat\Psi_0(\lambda)(\lambda+x)^{\frac{\alpha-\beta}{2}\sigma_3}
\lambda^{-\frac{\alpha+\beta}{2}\sigma_3}G_1^{-1}(\lambda+x)^{-\frac{\alpha-\beta}{2}\sigma_3}
\lambda^{-\frac{\alpha+\beta}{2}\sigma_3}\\
\times\quad
\begin{pmatrix}1& -c_0 J(\lb;x,\al,\bt)
\\0&1\end{pmatrix}E(\lambda)^{-1},\qquad \lb\in I'.
\end{multline}
Substituting (\ref{F1}) with $z=e^{i\pi}x/\lb$ into (\ref{F}), we obtain
the following representation
\be\label{J00}
J(\lb;x,\al,\bt)=
\frac{\lb^{\al+\bt}(\lb+x)^{\al-\bt}}{\sin\pi(\al+\bt)}-
x^{2\al}\frac{\Ga(\al+\bt)\Ga(1+\al-\bt)}{\pi\Ga(1+2\al)}
F(1,-2\al,1-\al-\bt,-\lb/x).
\ee
Therefore,
\be\label{J0}
\eqalign{
c_0 J(\lb;x,\al,\bt)=
-e^{2\pi i\al}\frac{\sin\pi(\al-\bt)}{\sin2\pi\al}
\lb^{\al+\bt}(\lb+x)^{\al-\bt}+F_0(\lb),\\
F_0(\lb)= -x^{2\al} e^{2\pi
i\al}\frac{\sin\pi(\al-\bt)}{\pi}\frac{\Ga(1+\al-\bt)\Ga(-2\al)}{\Ga(1-\al-\bt)}
F(1,-2\al,1-\al-\bt,-\lb/x).}
\ee
Since $F_0(\lb)$ is analytic close
to $\lb=0$, this representation gives an explicit expression for the
singularity of $J(\lb)$ at $0$.
(Note that the arguments of $\lb$, $\lb+x$ here are between $-\pi$ and $\pi$
as they originate from the hypergeometric function above, whereas
the arguments of  $\lb$, $\lb+x$ in the first row of (\ref{psip}) are between $0$ and $2\pi$.)
Substituting (\ref{J0}) into
(\ref{psip}) we see that the singularity cancels, and
$\widehat\Psi(\lambda)P(\lambda)^{-1}$ is analytic at $\lb=0$.

In order to analyze the singularity of $J$ at $\lb=-x$, apply first the
transformation of the hypergeometric
function between arguments $z$ and $1-z$ to (\ref{F1}) with $z=e^{i\pi}x/\lb$. We then obtain:
\be\label{F2}
\eqalign{
F(1,1+\al+\bt,2+2\al,-x/\lb)=\\
\frac{\pi e^{-2\pi
i\al}}{\sin\pi(\al-\bt)}\frac{\Ga(2+2\al)}{\Ga(1+\al+\bt)\Ga(1+\al-\bt)}
\left({\lb\over x}\right)^{1+\al+\bt}\left(1+{\lb\over x}\right)^{\al-\bt}\\
-\frac{1+2\al}{\al-\bt}{\lb\over x}F(1,-2\al,1-\al+\bt,1+\lb/x),}
\ee
and therefore by (\ref{F}),
\be\label{J10}
\eqalign{
J(\lb;x,\al,\bt)=
-\frac{e^{-2\pi i\al}\lb^{\al+\bt}(\lb+x)^{\al-\bt}}{\sin\pi(\al-\bt)}\\
+x^{2\al}\frac{\Ga(\al-\bt)\Ga(1+\al+\bt)}{\pi\Ga(1+2\al)}
F(1,-2\al,1-\al+\bt,1+\lb/x).}
\ee
Using (\ref{C0}), we obtain
\be\label{J11}
\eqalign{
c_0 J(\lb;x,\al,\bt)=
\frac{\sin\pi(\al+\bt)}{\sin2\pi\al}\lb^{\al+\bt}(\lb+x)^{\al-\bt}+F_1(\lb),\\
F_1(\lb)=x^{2\al}
e^{2\pi
i\al}\frac{\sin\pi(\al+\bt)}{\pi}\frac{\Ga(1+\al+\bt)\Ga(-2\al)}{\Ga(1-\al+\bt)}
F(1,-2\al,1-\al+\bt,1+\lb/x).}
\ee
This representation explicitly displays the singularity at $-x$ because $F_1(\lb)$
is analytic near $\lb=-x$.
In the same way as for $\lb=0$, we now obtain that
$\widehat\Psi(\lambda)P(\lambda)^{-1}$ is analytic at $\lb=-x$ as well.

For $\lambda$ at a fixed distance away of the origin, say for
$\lambda\in\partial U_\varepsilon$, it follows from (\ref{def P2}) that
\begin{equation}
P(\lambda;x)M(\lambda)^{-1}=I+\bigO(x)+\bigO(x^{1+2\alpha}),\qquad\mbox{ as
$x\searrow 0$.}
\end{equation}

\medskip

We now consider the case $2\alpha\in\mathbb Z$. We again set
\begin{equation}
P(\lambda)=P_j(\lambda),\qquad\mbox{ with $j=1$ for $\lambda$ in sector I',
$j=2$ in sector II', and so on,}
\end{equation}
but now with
\be\label{wtP}
P_j(\lambda)=\wt E(\lambda)\begin{pmatrix}1& c_2\wt J(\lb)
\\0&1\end{pmatrix}
(\lambda+x)^{\frac{\alpha-\beta}{2}\sigma_3}\lambda^{\frac{\alpha+\beta}{2}\sigma_3}
\begin{pmatrix}1&\ga(\lb)\\0&1\end{pmatrix}\wt G_j,
\ee
where
\be\label{C2}
c_2=-{1\over\pi}\sin\pi(\al+\bt)\sin\pi(\al-\bt),
\ee
\be
\wt J(\lb;x,\al,\bt)={1\over 2}\left(
\frac{\partial}{\partial\al}+\frac{\partial}{\partial\bt}\right)J(\lb;x,\al,\bt),
\ee
with $J(\lb)$ given by (\ref{F}).

Note that
\be\label{wtJjump}
\wt J_+(\lb)=\wt J_-(\lb)+ 2i|\lb|^{\al+\bt}|\lb+x|^{\al-\bt}\ln|\lb|,\qquad \lb\in
(-x,0).
\ee

A similar derivation to the one above shows that $P$ again satisfies the same jump
conditions as the ones for $\widehat\Psi$.
To analyze the structure of the singularities of $P(\lb)$ at $0$ and at $-x$, we
need, as above,
to find suitable expressions for $c_2\wt J(\lb)$ at these points.
Applying the differential operator ${1\over 2}(\partial/\partial\al+
\partial/\partial\bt)$
to (\ref{J00}) (it is convenient to write
$\lb^{\al+\bt}=e^{i\pi(\al+\bt)}(e^{-i\pi}\lb)^{\al+\bt}$ first)
and then taking $2\al\in\mathbb Z$ and using the fact that
\[
\frac{\sin\pi(\al-\bt)}{\sin\pi(\al+\bt)}=(-1)^{2\al+1},\qquad 2\al\in\mathbb Z,
\quad\bt\notin\mathbb Z,
\]
we obtain
\be\label{wtJ0}
\eqalign{
c_2\wt J(\lb;x,\al,\bt)=
\left[(-1)^{2\al+1}(\lb e^{-i\pi})^{\al+\bt}-
{1\over\pi}\sin\pi(\al-\bt)\lb^{\al+\bt}\ln(\lb e^{-i\pi})\right]
(\lb+x)^{\al-\bt}+\wt F_0(\lb)\\
\wt F_0(\lb)=-{c_2\over 2}\left(
\frac{\partial}{\partial\al}+\frac{\partial}{\partial\bt}\right)
\left[
x^{2\al}\frac{\Ga(\al+\bt)\Ga(1+\al-\bt)}{\pi\Ga(1+2\al)}
F(1,-2\al,1-\al-\bt,-\lb/x)\right].}
\ee
As above, this expression can be used to show that $\widehat\Psi(\lb)P(\lb)^{-1}$
has no singularity at $\lb=0$.
To analyze a neighborhood of $\lb=-x$, note that the first term on the r.h.s. of
(\ref{J10})
can be written as
\[
-\frac{e^{-\pi i(\al-\bt)}}{\sin\pi(\al-\bt)}
(e^{-i\pi}\lb)^{\al+\bt}(\lb+x)^{\al-\bt},
\]
and application of the operator ${1\over 2}(\partial/\partial\al+
\partial/\partial\bt)$
to the fraction in the above formula gives zero by antisymmetry in $\al$ and $\bt$.
We finally obtain from (\ref{J10})
\be\label{wtJ1}
\eqalign{
c_2\wt J(\lb;x,\al,\bt)=
{1\over\pi}\sin\pi(\al+\bt)e^{-2\pi\al}\lb^{\al+\bt}(\lb+x)^{\al-\bt}
\ln(\lb e^{-i\pi})+\wt F_1(\lb)\\
\wt F_1(\lb)={c_2\over 2}\left(
\frac{\partial}{\partial\al}+\frac{\partial}{\partial\bt}\right)
\left[
x^{2\al}\frac{\Ga(\al-\bt)\Ga(1+\al+\bt)}{\pi\Ga(1+2\al)}
F(1,-2\al,1-\al+\bt,1+\lb/x)\right],}
\ee
which can be used to see that $\widehat\Psi(\lb)P(\lb)^{-1}$ has no singularity at
$\lb=-x$.

It is easy to see that
\begin{equation}
P(\lambda;x)M(\lambda)^{-1}=I+\bigO(x\ln x),\qquad\mbox{ for $\lambda\in\partial U$,}
\end{equation}
as $x\searrow 0$.

Note that using an integral representation for the hypergeometric function,
we can, if $\Re(\al\pm\bt)>-1$, represent $J(\lb)$ and $\wt J(\lb)$ in the following
form,
which makes the jump conditions (\ref{Jjump}) and (\ref{wtJjump}) obvious:
\be\eqalign{
J(\lb)=\frac{1}{\pi}
\int_{-x}^0\frac{|\xi+x|^{\alpha-\beta}|\xi|^{\alpha+\beta}}{\xi-\lambda}d\xi,\qquad
\wt J(\lb)=\frac{1}{\pi}
\int_{-x}^0\frac{|\xi+x|^{\alpha-\beta}|\xi|^{\alpha+\beta}\ln|\xi|}
{\xi-\lambda}d\xi,\\
\Re(\al\pm\bt)>-1.}
\ee

\medskip

Now define
\begin{equation}\label{RPsi}
R(\lambda)=
\begin{cases}\widehat\Psi(\lambda)M(\lambda)^{-1},&\mbox{ for $\lambda\in\mathbb
C\setminus
U_\varepsilon$,}\\
\widehat\Psi(\lambda)P(\lambda)^{-1},&\mbox{ for $\lambda\in U_\varepsilon$.}
\end{cases}
\end{equation}

This function satisfies the following problem.
\subsubsection*{RH problem for $R$}
\begin{itemize}
\item[(a)] $R$ is analytic in $\mathbb C\setminus \partial U_\varepsilon$.
\item[(b)] The jump condition for $R$ is
\begin{equation}
R_+(\lambda)=R_-(\lambda)(I+e(\lb)), \qquad\mbox{ for $\lambda\in\partial U$,
with $e(\lb)=o(1)$ as $x\searrow 0$.}
\end{equation}
\item[(c)] $R(\lambda)=I+\bigO(\lambda^{-1})$ as $\lambda\to\infty$.
\end{itemize}
 For $x$ sufficiently small, this is a small-norm RH problem, and it follows that
the RH problem for $R$ is solvable, say for $0<x<\delta$. From the invertible
transformations $\Psi\mapsto \widehat\Psi\mapsto R$ and $\Psi\mapsto\Phi$, it
follows that the RH problems for $\Psi$ and $\Phi$ are solvable as well for
$0<x<\delta$. We also have $R(\lambda)=I+o(1)$ uniformly for $\lambda\in\mathbb
C\setminus \partial U$ as $x\searrow 0$. In particular this holds at infinity,
which means that
\begin{equation}
R(\lambda)=I+\bigO\left(\frac{1}{\lambda}\right), \qquad \mbox{ as $\lambda\to\infty$,}
\end{equation} uniformly for small $x$.
Tracing back the transformations $\Psi\mapsto \widehat\Psi\mapsto R$ and
$\Psi\mapsto \Phi$, we can conclude that
\begin{equation}\label{Phi infty 0}\Phi(\lambda; x)=I+\bigO(\lambda^{-1}),
\qquad\mbox{ as $\lambda\to\infty$, uniformly for $0<x<\delta$}.\end{equation}

Moreover, using (\ref{J0}) we obtain for $2\al\notin\mathbb Z$
\be\label{PhiPhi1}
\Phi(\frac{x}{2};x)\si_3\Phi(\frac{x}{2};x)^{-1}=
e^{-x\si_3/4}R(0)E(0)
\begin{pmatrix}
-1 & 2F_0(0)\cr
0 & 1
\end{pmatrix}
E(0)^{-1}R(0)^{-1}e^{x\si_3/4},
\ee
where
$F_0(\lb)$ is defined in (\ref{J0}).
Similarly,
\be\label{PhiPhi2}
\eqalign{
\Phi(-\frac{x}{2};x)\si_3\Phi(-\frac{x}{2};x)^{-1}=\\
e^{-x\si_3/4}R(-x)E(-x)
\begin{pmatrix}
1 & -2F_1(-x)\cr
0 & -1
\end{pmatrix}
E(-x)^{-1}R(-x)^{-1}e^{x\si_3/4},}
\ee
where
$F_1(\lb)$ is defined in (\ref{J11}).

Therefore, we have
\be\label{Phi0est}
\Phi(\pm\frac{x}{2};x)\si_3\Phi(\pm\frac{x}{2};x)^{-1}=\bigO(1)+\bigO(x^{2\al}),\qquad
x\searrow 0,\qquad 2\al\notin\mathbb Z.
\ee
Similarly, by (\ref{wtP}), (\ref{wtJ0}), (\ref{wtJ1}),
\be\label{Phi0estwt}
\Phi(\pm\frac{x}{2};x)\si_3\Phi(\pm\frac{x}{2};x)^{-1}=\bigO(1)+\bigO(x^{2\al})+
\bigO(x^{2\al}\ln x),\qquad
x\searrow 0,\qquad 2\al\in\mathbb Z.
\ee

We will now estimate $w(x)$ given by (\ref{defw}) as $x\searrow 0$. First, using the
connection
(\ref{def Phi}) between $\Phi$ and $\Psi$, we obtain
\begin{equation}\label{PhiPsi}
(\Phi^{-1}\Phi'_\lb)_{jj}= (\Psi^{-1}\Psi'_\lb)_{jj}-{(-1)^j\over 2}
\left(\frac{\al+\bt}{\lb-x/2}-\frac{\al-\bt}{\lb+x/2}+1\right),\qquad j=1,2.
\end{equation}
Expressing $\Psi^{-1}\Psi'_\lb=\widehat\Psi^{-1}\widehat\Psi'_\lb$, and using
(\ref{RPsi}), and the formulas for $P(\lb)$, we obtain after straightforward
calculations that
\be\label{Phit}\eqalign{
(\Phi^{-1}\Phi'_\lb)_{11}(x/2)=-{\al-\bt\over x}+(d_1+d_2x^{2\al})(1+\bigO(x)),\qquad
d_1=\frac{\al-\bt}{2\al},\\
d_2=\frac{\al-\bt}{1+2\al}\frac{\Gamma(1+\al+\bt)\Gamma(1+\al-\bt)}{\Gamma(1-\al+\bt)\Gamma(1-\al-\bt)}
\frac{\Gamma(-2\al)}{\Gamma(1+2\al)^2},\qquad 2\al\notin\mathbb Z.}
\ee
Similarly,
\be\label{Phi-t}\eqalign{
(\Phi^{-1}\Phi'_\lb)_{22}(-x/2)={\al+\bt\over x}+(\wt d_1+\wt
d_2x^{2\al})(1+\bigO(x)),\qquad
\wt d_1=-\frac{\al+\bt}{2\al},\\
\wt d_2=-\frac{\al+\bt}{1+2\al}\frac{\Gamma(1+\al+\bt)\Gamma(1+\al-\bt)}
{\Gamma(1-\al+\bt)\Gamma(1-\al-\bt)}
\frac{\Gamma(-2\al)}{\Gamma(1+2\al)^2},\qquad 2\al\notin\mathbb Z.}
\ee
Substituting these expressions into (\ref{defw}), we obtain as $x\searrow 0$:
\be\label{was}\eqalign{
w(x)=
{\al^2-\bt^2\over x}+
\frac{\al^2-\bt^2}{2\al}\\
\times\left\{1-
x^{2\al}\frac{\Gamma(1+\al+\bt)\Gamma(1+\al-\bt)}{\Gamma(1-\al+\bt)\Gamma(1-\al-\bt)}
\frac{\Gamma(1-2\al)}{\Gamma(1+2\al)^2}\frac{1}{1+2\al}\right\}(1+\bigO(x)),\qquad
2\al\notin\mathbb Z.}
\ee
Similarly, we verify using (\ref{wtP}), (\ref{wtJ0}), (\ref{wtJ1}) that as
$x\searrow 0$,
\be\label{waswt}
w(x)=
{\al^2-\bt^2\over x}+\bigO(1)+\bigO(x^{2\al})+
\bigO(x^{2\al}\ln x),\qquad 2\al\in\mathbb Z.
\ee

\subsection{Differential system for $\Psi$}\label{section 43}
So far we know that there exist $\delta, M>0$ such that the RH
problems for $\Psi$ and $\Phi$ are solvable for $x>M$ and for
$0<x<\delta$. We will derive differential equations for $\Psi$ with
respect to $x$ and $\zeta$. This will lead to the Painlev\'e V
equation and will help us to find an identity for the function $w$
given by (\ref{defw}) in terms of $v$. Here we follow similar lines
as in \cite[Section 5.4]{FIKN}.

\medskip

From the RH conditions for $\Psi$, it follows that, for any $x$ for which the RH
problem is solvable, the (matrix) function
$A(\zeta;x)=\Psi_\zeta(\zeta;x)\Psi^{-1}(\zeta;x)$ is a rational function in $\zeta$
with simple poles at $0$ and $1$. Indeed, $A$ is meromorphic because $\Psi$ has
constant jump matrices, $A$ is bounded at infinity because of (\ref{RHP Psi: c}),
and has simple poles at $0$ and $1$ because of (\ref{def Psi0}). Similarly,
$B(\zeta;x)=\Psi_x(\zeta;x)\Psi^{-1}(\zeta;x)$ is a polynomial of degree $1$ in
$\zeta$. It follows that $\Psi$ satisfies a linear differential system of the form
   \begin{align}
&\label{PsiA}\Psi_\zeta(\zeta;x)=\left[A_\infty(x)+\frac{A_0(x)}{\zeta}+\frac{A_1(x)}{\zeta-1}\right]\Psi(\zeta;x),\\
&\label{PsiB}\Psi_x(\zeta;x)=\left[B_1(x)\zeta+B_0(x)\right]\Psi(\zeta;x).
\end{align}
Substituting the large $\zeta$-expansion (\ref{RHP Psi: c}) for $\Psi$ into
(\ref{PsiA}) and (\ref{PsiB}), we can express the coefficient matrices
$A_\infty$, $A_0$, $A_1$, $B_0$, and $B_1$ explicitly in terms of the entries of
$C_1$ and $C_2$:
\begin{align}
&\label{A} A_\infty=-\frac{x}{2}\sigma_3, \\
&\label{A0}A_0=\begin{pmatrix}-\beta+q+xrt&-2\beta r-xh+xrq+r+xr\\2\beta
t+xj-xtq+t-xt&\beta-q-xrt\end{pmatrix},\\
&\label{A1}A_1=\begin{pmatrix}-q-xrt&2\beta r+xh-xrq-r\\-2\beta
t-xj+xtq-t&q+xrt\end{pmatrix},\\
& B_1=-\frac{1}{2}\sigma_3, \\ &\label{B}B_0=\begin{pmatrix}0&
r\\ -t
&0\end{pmatrix},
\end{align}
where $q=q(x)$, $r=r(x)$, $t=t(x)$, $h=h(x)$, and $j=j(x)$ are given by
\begin{equation}\label{def C1}
C_1(x)=\begin{pmatrix}q(x)&r(x)\\t(x)&-q(x)
\end{pmatrix},\qquad C_2(x)=\begin{pmatrix}*&h(x)\\j(x)&*
\end{pmatrix}
\end{equation}
(note that the trace of $C_1$ must be zero since the determinant of $\Psi$ is equal
to $1$).
Equating the $\bigO(1/\zeta)$-terms in $\Psi_x=(B_1\zeta+B_0)\Psi$ gives the identities
\begin{align}&q'(x)=r(x)t(x),\label{qrt}\\
&h(x)=-r'(x)+r(x)q(x),\\
&j(x)=t'(x)+t(x)q(x).
\end{align}
Furthermore, equating the mixed derivatives $\Psi_{x\zeta}=\Psi_{\zeta x}$ leads to
the compatibility condition
\begin{equation}\label{compatibility}
A_x-B_\zeta+[A,B]=0,\qquad [A,B]=AB-BA.
\end{equation}
Let us follow \cite{FokasMuganZhou, FIKN} and write
\begin{align}
&\label{defv}v(x)=\frac{\alpha+\beta}{2}-q(x)-xq'(x)=\frac{\alpha+\beta}{2}-q(x)-xr(x)t(x), \\
&\label{def y}y(x)=\frac{v(x)}{(-2\beta-1) t(x)-xt'(x)},\\
&\label{def u}u(x)=1+\frac{xt}{(2\beta+1-x)t(x)+xt'(x)}.
\end{align}
Using Proposition \ref{prop Psi} one shows as in \cite{FIKN} that
$\det A_0=-\frac{(\alpha-\beta)^2}{4}$ and $\det A_1=-\frac{(\alpha+\beta)^2}{4}$.
It then follows that the matrices $A_0$, $A_1$, and $B_0$ can be written in the form
(the elements (11), (22), and (21) of $A_0$ and $A_1$ are easy to verify directly,
and (12) follows from the expression for the determinant):
\begin{align}
&A_0=\begin{pmatrix}-v+\frac{\alpha-\beta}{2}&
uy(v-\alpha+\beta)\\-\frac{v}{uy}&v-\frac{\alpha-\beta}{2}\end{pmatrix},\label{2A0}
\\ &A_1=\begin{pmatrix}v-\frac{\alpha+\beta}{2}&-y(v-\alpha-\beta)\\
\frac{v}{y} &-v+\frac{\alpha+\beta}{2}\end{pmatrix},\label{2A1}
\\
&B_0=\frac{1}{x}\begin{pmatrix}0&
-y[v-\alpha-\beta-u(v-\alpha+\beta)]\\ \frac{1}{y}[v-\frac{v}{u}]
&0\end{pmatrix}.\label{2B}
\end{align}
Writing the compatibility condition (\ref{compatibility}) in terms of the functions
$u, v, y$, one verifies that
$u$, $v$, and $y$ solve the system of ODEs
\begin{align}\label{system uvy}
&xu_x=xu-2v(u-1)^2+(u-1)[(\alpha-\beta)u-\beta-\alpha],\\
&\label{system uvy2}xv_x=uv[v-\alpha+\beta]-\frac{v}{u}(v-\beta-\alpha),\\
&\label{system uvy3}xy_x=y\left\{-2v+\alpha
+\beta+u[v-\alpha+\beta]+\frac{v}{u}-x\right\},
\end{align}
which is part of the content of Theorem \ref{theorem: Painleve}
(iii). Eliminating $v$ from the first two equations, one shows that
$u$ solves the Painlev\'e V equation (\ref{PVintro})-(\ref{ABCD}).

Define
\begin{equation}\label{sigmaq}
\sigma(x) = xq(x) -\frac{\alpha + \beta}{2}x.
\end{equation}
It follows from (\ref{defv}) that
\begin{equation}\label{sigmaq0}
\sigma' = -v,
\end{equation}
and therefore, by (\ref{system uvy2}),
\begin{equation}\label{sigmaq2}
- x\sigma'' = uv(v-\alpha+\beta)-\frac{v}{u}(v-\beta-\alpha).
\end{equation}
Moreover, in view of (\ref{qrt}), we have that
\[
\sigma -x\sigma' = -x^2q' = -x^2rt \equiv x^2(B_0)_{12}(B_0)_{21}.
\]
This equation can be rewritten with the help of (\ref{2B}) as
\begin{multline}
\sigma -x\sigma' = -\Bigl(v-\alpha -\beta  -u(v-\alpha
+\beta)\Bigr)\left(v - \frac{v}{u}\right)\\ \label{sigmaq3}
=uv(v-\alpha+\beta)+\frac{v}{u}(v-\beta-\alpha) -2v^2 +2\alpha v.
\end{multline}
Using (\ref{sigmaq0}),  (\ref{sigmaq2}), and  (\ref{sigmaq3}), we
can check directly that the function $\sigma(x)$ satisfies the
$\sigma$-form of the fifth Painlev\'e equation (\ref{sigma5}).

The system (\ref{PsiA})-(\ref{PsiB}) is the Lax pair associated with Painlev\'e V.
Since the RH problem for $\Psi(\zeta;x,\alpha,\beta)$ is solvable for $0<x<\delta$
and for $x>M$, the Lax matrices $A_0(x; \alpha, \beta)$, $A_1(x; \alpha, \beta)$,
and $B_0(x; \alpha, \beta)$ exist for those values of $x$.
However, the system (\ref{system uvy})-(\ref{system uvy3}) has solutions which are
meromorphic in $\mathbb C\setminus\{0\}$ with a cut from zero to infinity, which
implies that $A_0$, $A_1$, and $B_0$ exist for all but (possibly) a finite number of
positive $x$-values. Using appropriately normalized solutions to
(\ref{PsiA})-(\ref{PsiB}), the RH solution $\Psi$ can also be constructed for all
but possibly a finite number of positive $x$-values \cite{FIKN}.
This proves Theorem \ref{theorem: Painleve} (i) and the equivalent statement for
$\Phi$, Proposition \ref{prop Phi} (i).
Furthermore, the differentiability of $\Psi$ with respect to $x$, see (\ref{PsiB}),
implies that the asymptotic condition (\ref{RHP Psi: c}), and thus also (\ref{RHP
Ph: c}), holds uniformly as long as $\delta\leq x\leq M$ if $x$ remains bounded away
from the set of $x$-values for which the RH problem is not solvable. Together with
(\ref{Phi infty}) and (\ref{Phi infty 0}), this proves Proposition \ref{prop Phi}
(iii).


\begin{remark}
The functions $u, v, y$ appearing in (\ref{2A0})--(\ref{2B}) are particular
solutions to the system (\ref{system uvy})--(\ref{system uvy3}). Other solutions can
be obtained by considering RH problems for $\Psi$ with modified jump matrices and
modified behavior near $0$ and $1$, corresponding to different monodromy data, see
\cite{AndreevKitaev}.
\end{remark}
%

\begin{remark}\label{remark notation}
The RH problem for $\Psi$ is not the standard RH problem related to the fifth
Painlev\'e equation. In \cite{FIKN, FokasMuganZhou}, a RH problem  was posed on a
contour $U_0\cup U_1\cup\widehat\Gamma$, where $U_0$ and $U_1$ are small circles
surrounding $0$, and $1$, and where $\widehat\Gamma=\mathbb R\setminus
(\overline{U_0}\cup \overline{U_1})$. The equivalence of
a particular case of this RH problem with ours can be verified directly using
Proposition \ref{prop Psi}.
In order to avoid confusion with the notations in \cite{FIKN}, we note that the
system (\ref{system uvy})-(\ref{system uvy3}) is written with parameters $\theta_0$,
$\theta_1$, and $\theta_\infty$ in \cite{FIKN}, which in our setting are given by
\begin{equation}
\theta_0=-\beta-\alpha, \qquad \theta_1=\alpha-\beta, \qquad \theta_\infty=2\beta.
\end{equation}
\end{remark}

\begin{proposition}\label{prop eigenvectors}
Set
\begin{equation}\label{def a}
a(\zeta;x)=\left(\Psi(\zeta;x)\sigma_3\Psi^{-1}(\zeta;x)\right)_{11}.
\end{equation}
Then the identities
\begin{align}
&\label{ida1}\frac{\alpha-\beta}{2}a(0;x)=A_{0,11}=-v(x)+\frac{\alpha-\beta}{2},\\
&\label{ida2}\frac{\alpha+\beta}{2}a(1;x)=-A_{1,11}=-v(x)+\frac{\alpha+\beta}{2},
\end{align}
hold, with $v$ defined as before by (\ref{def v intro}).
\end{proposition}
\begin{proof}
Substituting $\Psi$ expressed from (\ref{def Psi0}) into the differential equation
\[\Psi_{\zeta}\Psi^{-1}=A_\infty+\frac{A_0}{\zeta}+\frac{A_1}{\zeta-1},\]
and comparing the residue of the left- and right-hand side at $0$ leads to an
expression for $\Psi(\zeta)\si_3\Psi(\zeta)^{-1}$ as $\zeta\to 0$, in terms of
$A_0$. By
(\ref{2A0}), this gives the
first identity. Comparing the residues at $1$ gives the second identity.
\end{proof}

\begin{proposition}\label{prop w v}
Let $w$ be defined by (\ref{defw}). Then
\begin{align}\label{def wv}
&v(x)=-(xw(x))',\\
&\sigma(x)=x w(x),\label{sw}\\
&\sigma(x)=\int_x^{+\infty}v(\xi)d\xi.\label{sint}
\end{align}
\end{proposition}
\begin{proof}
It follows from Proposition \ref{prop Psi} that $\Psi$ can be written in the form
\begin{equation}\label{PsiEF}
\Psi(\zeta)=E(\zeta)\zeta^{\frac{\alpha-\beta}{2}\sigma_3}, \qquad
\Psi(\zeta)=F(\zeta)(\zeta-1)^{-\frac{\alpha+\beta}{2}\sigma_3},
\end{equation}
with $E$ analytic near $0$ and $F$ analytic near $1$. Let us write
\begin{align}
&\label{EF1}E(\zeta)=E_0(I+E_1\zeta+\bigO(\zeta^2)),&\mbox{ as $\zeta\to 0$,}\\
&\label{EF2}F(\zeta)=F_0(I+F_1(\zeta-1)+\bigO((\zeta-1)^2)),&\mbox{ as $\zeta\to 1$.}
\end{align}
Substituting (\ref{PsiEF}) and (\ref{EF1})--(\ref{EF2}) into (\ref{PsiB}),
we obtain the identities
\begin{align*}
&E_{0,x}'=B_0E_0, & E_{1,x}'=E_0^{-1}B_1E_0,\\
&F_{0,x}'=B_0F_0, & F_{1,x}'=F_0^{-1}B_1F_0,
\end{align*}
which imply by Proposition \ref{prop eigenvectors}, in particular, that
\begin{equation}
\label{diffEF}
E_{1,22}'(x)=(E_0^{-1}B_1E_0)_{22}=(E_0B_1E_0^{-1})_{22}=\frac{1}{2}a(0;x),
\qquad F_{1,11}'(x)=-\frac{1}{2}a(1;x).
\end{equation}
On the other hand recalling equation (\ref{PhiPsi}), we obtain
\begin{align}
&\left(\Phi^{-1}(-\frac{x}{2};x)\Phi'_\lambda(-\frac{x}{2};x)\right)_{22}=-\frac{1}{2}+\frac{\alpha+\beta}{2x}+\frac{1}{x}E_{1,22},\\
&\left(\Phi^{-1}(\frac{x}{2};x)\Phi'_\lambda(\frac{x}{2};x)\right)_{11}=\frac{1}{2}-\frac{\alpha-\beta}{2x}+\frac{1}{x}F_{1,11}.
\end{align}
From (\ref{defw}), it follows that
\begin{equation}
w(x)=-\frac{\alpha}{2}+\frac{\alpha^2-\beta^2}{2x}+\frac{\alpha-\beta}{2x}E_{1,22}-\frac{\alpha+\beta}{2x}F_{1,11},
\end{equation}
and by (\ref{diffEF}) together with Proposition \ref{prop eigenvectors} we obtain
$-(xw(x))'=v(x)$.
From (\ref{sigmaq0}) and (\ref{def wv}), it follows that
$\sigma(x) = xw(x) + \mbox{constant}$, where $\sigma$ is defined by (\ref{sigmaq}).
To determine the constant, note first that, as follows from (\ref{104}) and (\ref{vest1}),
$q(x)=C_{1,11}\to (\al+\bt)/2$ as $x\to+\infty$, and hence, $\sigma(x)\to 0$ as $x\to+\infty$.
On the other hand, as follows from (\ref{as+0}), (\ref{as+1}),
we have $xw(x)\to 0$ as $x\to+\infty$. Hence the constant in question is zero, and
we obtain (\ref{sw}). Equation (\ref{sint}) is obtained similarly.
\end{proof}

Combining (\ref{def wv}), (\ref{was}), (\ref{waswt}), and (\ref{vest6}), we obtain
(\ref{vestintro}). The expressions (\ref{was}), (\ref{waswt}), (\ref{sw}), and
(\ref{sint}) imply (\ref{intv}).

\begin{proposition}\label{prop v real}
Let $v$ be defined by (\ref{def v intro}). Then $v(x)$ is real for
$x>0$ if $\Im\alpha=0$, $\alpha>-\frac{1}{2}$, and $\Re\beta=0$.
\end{proposition}
\begin{proof}
Suppose that $\alpha>-\frac{1}{2}$, $\Re\beta=0$, and that $\Psi(\zeta;x)$ is a
solution to the RH problem for $\Psi$ given in Section \ref{PVsection}. Then it is
straightforward to verify that the function $\widehat\Psi$ defined by
\[
\widehat\Psi(\zeta):=\sigma_1\overline{\Psi(-\overline{(\zeta-1/2)})}
\sigma_1e^{\pm\pi i\beta\sigma_3}, \qquad \mbox{ if $\pm\Im\zeta>0$,}\]
with $\sigma_1=\begin{pmatrix}0&1\\1&0\end{pmatrix}$, solves the RH problem for
$\Psi(\zeta+1/2)$ for real $x$ up to a constant factor.
Therefore, by uniqueness,
\begin{equation}
\Psi(\zeta+1/2)=C(x)^{\si_3}\widehat\Psi(\zeta),
\end{equation}
where $C(x)$ is independent of $\zeta$. By (\ref{def a}), it follows that
$a(0;x)=\overline{a(1;x)}$.
Subtracting the complex conjugate of (\ref{ida2}) from (\ref{ida1}), we
conclude that $v(x)=\overline{v(x)}$.
\end{proof}

    \subsection{Solvability of the RH problem for $\Psi$}\label{solvability}
    In this section, we will prove Theorem \ref{theorem: Painleve} (ii) and
Proposition \ref{prop Phi} (ii):
    we will prove that the RH problem for $\Psi$ is solvable for all positive values
of $x$ if $\Re\beta=0$
and $\Im\alpha=0$, $\alpha>-\frac{1}{2}$.

        \subsubsection{Vanishing lemma for Painlev\'e V}
        For a general class of RH problems, it is known that solvability of a RH
problem is equivalent to the triviality of a homogeneous version of the RH
problem \cite{FIKN, FokasZhou, KMM}. For the case of Painlev\'e V, this has
been used in \cite{FokasMuganZhou} for a slightly different but equivalent
RH problem
        (cf.\ Remark \ref{remark notation}). In our case a sufficient (and
necessary) condition to prove the solvability of the RH problem for $\Psi$
is given by the following so-called vanishing lemma.
        \begin{lemma}{\bf (Vanishing lemma for Painlev\'e V)}\ Let $x>0$,
$\Im\alpha=0$, $\Re\beta=0$, and suppose that $\Psi_0$ satisfies the RH
conditions (a), (b), (d0), and (d1) of the RH problem for $\Psi$, with
condition (c) replaced by the homogeneous asymptotic condition
        \begin{equation}
    \Psi_0(\zeta)e^{\frac{x}{2}\zeta\sigma_3}=\bigO(\zeta^{-1})
    , \qquad \mbox{ as $\zeta\to\infty$}.
    \end{equation}
    Then it follows that $\Psi_0 \equiv 0$.\label{vanishing lemma}
        \end{lemma}
        \begin{remark}
        A vanishing lemma was proven in \cite{FokasMuganZhou} for a family of
solutions to the
        system (\ref{system uvy})--(\ref{system uvy3}). Our solution, however, is
not contained in this family, and the vanishing lemma requires a different
proof in our case. For the proof of the vanishing lemma, we follow similar
lines as in \cite[Section 5.3]{DKMVZ2}.
        \end{remark}
        \begin{varproof}{\bf of Lemma \ref{vanishing lemma}.}
        Suppose we have a solution $\Psi_0$ to the homogeneous RH problem. We will
then prove that $\Psi_0\equiv 0$.
        Let us first define a function $M$ as follows,
        \begin{align*}
        &M(\zeta)=\Psi_0(\zeta+\frac{1}{2})e^{\frac{x}{2}\zeta\sigma_3},&&\mbox{ if
$\frac{3\pi}{4}<\arg\zeta<\frac{5\pi}{4}$,}\\
        &M(\zeta)=\Psi_0(\zeta+\frac{1}{2})e^{-\pi
i\beta\sigma_3}e^{\frac{x}{2}\zeta\sigma_3},&&\mbox{ for
$0<\arg\zeta<\frac{\pi}{4}$,}\\
        &M(\zeta)=\Psi_0(\zeta+\frac{1}{2})e^{\pi
i\beta\sigma_3}e^{\frac{x}{2}\zeta\sigma_3},&&\mbox{ for
$\frac{7\pi}{4}<\arg\zeta<2\pi$}.
        \end{align*} In the remaining regions where $\frac{\pi}{4}<\arg\zeta<\frac{3\pi}{4}$ or $\frac{5\pi}{4}<\arg\zeta<\frac{7\pi}{4}$
        we define
        $M$ as the analytic continuation of $M$ from the other sectors in such a way that $M$ has
jumps only on the imaginary and on a part of the real axis:
\begin{align*}
        &M(\zeta)=\Psi_0(\zeta+\frac{1}{2})\begin{pmatrix}1&-e^{\pi i(\alpha-\beta)}\\0&1\end{pmatrix}e^{-\pi
i\beta\sigma_3} e^{\frac{x}{2}\zeta\sigma_3},&&
        \mbox{ for
$\frac{\pi}{4}<\arg\zeta<\frac{\pi}{2}$},\\
&M(\zeta)=\Psi_0(\zeta+\frac{1}{2})\begin{pmatrix}1&0\\-e^{-\pi
i(\alpha-\beta)}&1\end{pmatrix}e^{\frac{x}{2}\zeta\sigma_3},&&\mbox{
for
$\frac{\pi}{2}<\arg\zeta<\frac{3\pi}{4}$},\\
        &M(\zeta)=\Psi_0(\zeta+\frac{1}{2})\begin{pmatrix}1&0\\-e^{\pi i(\alpha-\beta)}&1\end{pmatrix}e^{\frac{x}{2}\zeta\sigma_3},&&\mbox{ for
$\frac{5\pi}{4}<\arg\zeta<\frac{3\pi}{2}$},\\
        &M(\zeta)=\Psi_0(\zeta+\frac{1}{2})\begin{pmatrix}1&-e^{-\pi i(\alpha-\beta)}\\0&1\end{pmatrix}e^{\pi
i\beta\sigma_3}e^{\frac{x}{2}\zeta\sigma_3},&&\mbox{ for
$\frac{3\pi}{2}<\arg\zeta<\frac{7\pi}{4}$}.
        \end{align*}

        Then $M$ satisfies the following RH conditions.
        \subsubsection*{RH problem for $M$}
        \begin{itemize}
        \item[(a)] $M$ is analytic in $\mathbb C\setminus (i\mathbb R\cup
[-\frac{1}{2}, \frac{1}{2}])$.
        \item[(b)] $M$ satisfies the following jump conditions on the contour
$(i\mathbb R\cup (-\frac{1}{2}, \frac{1}{2}))$, with $i\mathbb R$ oriented
upwards and $(-\frac{1}{2}, \frac{1}{2})$ oriented from left to right,
        \begin{align}\label{RHP M: b1}
        &M_+(\zeta)=M_-(\zeta)V_1(\zeta),
        &&\mbox{ as $\zeta\in(0,+i\infty)$},\\
        &\label{RHP M: b2}M_+(\zeta)=M_-(\zeta)V_2(\zeta),
        &&\mbox{ as $\zeta\in(-i\infty,0)$},\\
        &\label{RHP M: b3}M_+(\zeta)=M_-(\zeta)e^{-\pi i(\alpha-\beta)\sigma_3},
        &&\mbox{ as $\zeta\in(-\frac{1}{2}, 0)$},\\
        &\label{RHP M: b4}M_+(\zeta)=M_-(\zeta)e^{-\pi i(\alpha+\beta)\sigma_3},
        &&\mbox{ as $\zeta\in(0, \frac{1}{2})$},
        \end{align}
        with
        \begin{align}
        &\label{V1}V_1(\zeta)={\small\begin{pmatrix}0&e^{\pi
i\alpha}e^{-x\zeta}\\-e^{-\pi i\alpha}e^{x\zeta}&e^{-\pi
i\beta}\end{pmatrix}},\\
    &\label{V2}V_2(\zeta)={\small \begin{pmatrix}0&e^{-\pi
i\alpha}e^{-x\zeta}\\-e^{\pi i\alpha}e^{x\zeta}&e^{\pi i\beta}\end{pmatrix}},
        \end{align}
        \item[(c)] For fixed $x>0$,
        \begin{equation}M(\zeta)=\bigO(\zeta^{-1}),\qquad\mbox{ as
$\zeta\to\infty$}.\end{equation}
        \item[(d0)] As $\zeta\to -\frac{1}{2}$,
    \begin{equation}\label{RHP M: d0}
    M(\zeta)=\bigO\begin{pmatrix}|\zeta+\frac{1}{2}|^{\frac{\alpha}{2}}&|\zeta+\frac{1}{2}|^{-\frac{\alpha}{2}}\\
    |\zeta+\frac{1}{2}|^{\frac{\alpha}{2}}&|\zeta+\frac{1}{2}|^{-\frac{\alpha}{2}}\end{pmatrix}.
    \end{equation}
    \item[(d1)] As $\zeta\to +\frac{1}{2}$,
    \begin{equation}\label{RHP M: d1}
    M(\zeta)=\bigO\begin{pmatrix}|\zeta-\frac{1}{2}|^{-\frac{\alpha}{2}}&|\zeta-\frac{1}{2}|^{\frac{\alpha}{2}}\\
    |\zeta-\frac{1}{2}|^{-\frac{\alpha}{2}}&|\zeta-\frac{1}{2}|^{\frac{\alpha}{2}}\end{pmatrix}.
    \end{equation}
    \end{itemize}

    \medskip

    Let us now define a function $H(\zeta)$ in terms of $M$ and its Hermitian
conjugate as follows:
        \begin{equation}
        H(\zeta)=M(\zeta)M^*(-\overline{\zeta}).
        \end{equation}
        Because of the condition (c) of the RH problem for $M$, we have that
$H(\zeta)=\bigO(\zeta^{-2})$ as $\zeta\to\infty$. Furthermore, using the
jump condition (\ref{RHP M: b3})--(\ref{RHP M: b4}) for $M$,
we obtain that $H$ has no jump across $(0,\frac{1}{2})$.
(This is only true if $\Re\beta=0$, $\Im\alpha=0$.)
        Therefore, $H$ is meromorphic for $\Re\zeta>0$, with an isolated singularity
at $\frac{1}{2}$, which is removable because
        of  (\ref{RHP M: d0}) and (\ref{RHP M: d1}).
        Using Cauchy's theorem, we then have
        \begin{equation}\label{HHC}
        \int_{-i\infty}^{+i\infty}H_-(\zeta)d\zeta=0,\qquad
\int_{-i\infty}^{+i\infty}H_-^*(\zeta)d\zeta=0.
        \end{equation}
        Because of the jump conditions for $M$, the first integral implies that
        \begin{equation}
        \int_{-i\infty}^{0}M_-(\zeta)V_2^*(\zeta)M_-^*(\zeta)d\zeta+
        \int_{0}^{+i\infty}M_-(\zeta)V_1^*(\zeta)M_-^*(\zeta)d\zeta=0.
        \end{equation}
        Summing up this expression and the one obtained from the second integral in
(\ref{HHC}),
we find, using (\ref{V1}), (\ref{V2}) and the fact that $x$ is real,
        \begin{equation}\label{symmetry}
        \int_{-i\infty}^0M_-(\zeta)\begin{pmatrix}0&0\\0&2e^{\pi
i\beta}\end{pmatrix}M_-^*(\zeta)d\zeta +
\int_0^{i\infty}M_-(\zeta)\begin{pmatrix}0&0\\0&2e^{-\pi
i\beta}\end{pmatrix}M_-^*(\zeta)d\zeta=0.
        \end{equation}
        Since $\Re\beta=0$, it follows immediately that the second column of $M_-$
is identically zero on $i\mathbb R\setminus\{0\}$. From the jump conditions
(\ref{RHP M: b1})-(\ref{RHP M: b2}), it then follows that the first column
of $M_+$ is zero on $i\mathbb R\setminus\{0\}$ as well. Therefore, we have
that $M_{j2}(\zeta)=0$ for $\Re\zeta>0$, and $M_{j1}(\zeta)=0$ for
$\Re\zeta<0$.
        Let us now define
        \begin{equation}
        g_j(\zeta)=\begin{cases}
        M_{j2}(\zeta), &\mbox{ as $\Re\zeta<0$},\\
        M_{j1}(\zeta), &\mbox{ as $\Re\zeta>0$},
        \end{cases}
        \end{equation}
        so that $g_j$ is analytic in $\mathbb C\setminus(i\mathbb R\cup
[-\frac{1}{2}, \frac{1}{2}])$.
        Furthermore, $g_j$ is bounded except near $\pm\frac{1}{2}$.
        On $i\mathbb R$, $g$ has the following jump relation,
        \begin{equation}
        g_{j,+}(\zeta)=g_{j,-}(\zeta)\times
        \begin{cases}
        e^{\pi i\alpha}e^{-x\zeta}, &\mbox{ as $\zeta\in(0,+i\infty)$,}\\
        e^{-\pi i\alpha}e^{-x\zeta}, &\mbox{ as $\zeta\in(-i\infty,0)$.}
        \end{cases}
        \end{equation}
        Now we write $\widehat g$ for the analytic continuation of $g$ from the left
half plane to $\mathbb C\setminus[-\frac{1}{2},+\infty)$,
        \begin{equation}
        \widehat g(\zeta)=\begin{cases}
        g(\zeta), &\mbox{ as $\Re\zeta<0$,}\\
        g(\zeta)e^{\pi i\alpha}e^{-x\zeta}, &\mbox{ as $\Re\zeta>0$, $\Im\zeta>0$,}\\
        g(\zeta)e^{-\pi i\alpha}e^{-x\zeta}, &\mbox{ as $\Re\zeta>0$, $\Im\zeta<0$.}
        \end{cases}
        \end{equation}
        Set
        \begin{equation}
        h(\zeta)=\widehat g(-(\zeta+1)^{3/2}).
        \end{equation}
        It is now easy to verify that $h$ is analytic and bounded for $\Re\zeta\geq
0$, and that
        $h(\zeta)=\bigO(e^{-x|\zeta|})$ for $\zeta\to\pm i\infty$. By Carlson's
theorem, this implies that $h\equiv 0$ if $x>0$. Tracing back the previous
steps, it follows that $g\equiv 0$, $M\equiv 0$, and $\Psi_0\equiv 0$, which
proves the vanishing lemma.
        \end{varproof}
        \begin{remark}
        The proof of the vanishing lemma does not apply if either $\alpha$ is not
real or $\beta$ is not purely imaginary. The first failure is that the
function $H$ would not be analytic across $(0,\frac{1}{2})$ in this case. A
further problem in the proof would be that the matrices $V_1+V_1^*$ and
$V_2+V_2^*$ lose their symmetry, which results in non-zero off-diagonal
entries in (\ref{symmetry}). It is of course possible that the vanishing
lemma can be proven in a different way. Another possibility is that, given
$\alpha$ and $\beta$, the RH problem is not solvable for certain isolated
values of $x$.
        \end{remark}

\section{Asymptotics for Toeplitz determinants}\label{section: integrate}
Using the identities of Proposition \ref{prop eigenvectors} and the Fourier
representation
for $V(z)$, we can rewrite (\ref{di-interim}) in the form, with $x=2nt$,
\be\label{di-interim2}\eqalign{
\frac{d}{dt}\ln D_n=
(\al+\bt)n-(\al^2-\bt^2)
{e^{-t}\over\sinh t}+
(\al-\bt)\sum_{k=1}^\infty k V_k e^{-kt}+(\al+\bt)\sum_{k=1}^\infty k V_{-k} e^{-kt}\\
+\frac{1}{t}\sigma(x) -v(x)\left\{ \al+\al \left({1\over t}-{e^{-t}\over\sinh
t}\right) +2\sum_{k=1}^\infty k(V_k+V_{-k})\cosh(kt)
\right\}+O(1/n)\wt\Phi(x).} \ee

The expressions (\ref{was}), (\ref{waswt}), (\ref{sw}) yield
the $x\to 0$ expansion for $\sigma$ in (\ref{westintro});
and the expressions (\ref{vest6}), (\ref{sint}) imply the
$x\to +\infty$ expansion in (\ref{westintro}).

Because of the uniformity property of the error term in (\ref{di-interim2}),
the integration of this identity from $\ep>0$ to some $t<t_0$ gives
uniformly for any $0<\ep<t$,
\be\label{di-interim3}\eqalign{
\ln D_n(t)=\ln D_n(\ep)+(\al+\bt)n(t-\ep)+
\sum_{k=1}^{\infty}k\left[V_k-(\alpha+\beta)\frac{e^{-tk}}{k}\right]
\left[V_{-k}-(\alpha-\beta)\frac{e^{-tk}}{k}\right]\\
-\sum_{k=1}^{\infty}k V_k V_{-k}
+(\al-\bt)\sum_{k=1}^\infty V_k e^{-k\ep}+(\al+\bt)\sum_{k=1}^\infty V_{-k} e^{-k\ep}\\
+\left[
\int_{2n\ep}^{2nt}\frac{\sigma(x)}{x}dx+(\al^2-\bt^2)\ln\{n(1-e^{-2\ep})\}\right]-
(\al^2-\bt^2)\ln n + R_n(t)+\bigO(1/n);\\
R_n(t)=-\int_{\ep}^t v(2nt)\left\{ \al+\al
\left({1\over t}-{e^{-t}\over\sinh t}\right)
+2\sum_{k=1}^\infty k(V_k+V_{-k})\cosh(kt)\right\}dt.}
\ee

If $\al$ is real and $\bt$ is imaginary, we can take as a path of integration the
interval $[\ep,t]$
of the real line as, according to Section \ref{solvability}, the functions $\sigma(x)$
and $v(x)$
are real analytic for positive $x$. The estimates (\ref{westintro}),
(\ref{vestintro}) ensure integrability
at $x=0$ and $x=+\infty$. In particular, the term in the square brackets in
(\ref{di-interim3})
converges if $\ep\to 0$.

For arbitrary $\bt$, $\Re\al>-1/2$, we can choose a path of integration
and the end-point $t$ to avoid possible singular points $\{x_1,\dots,x_k\}$.
The estimates (\ref{westintro}), (\ref{vestintro}) were obtained above for positive
$x$.
The restriction to real $x$ was only imposed for
simplicity of notation. In fact, it is easy to verify that the estimates
(\ref{westintro}), (\ref{vestintro})
hold for any path to zero and infinity within a sector
$-\pi/2+\de<\arg x<\pi/2+\de$, $0<\de<\pi/2$.

We have for $R_n(t)$ in (\ref{di-interim3}):
\be
|R_n(t)|< C \int_0^t |v(2nu)|du=\bigO(1/n),\qquad n\to\infty,\quad 0<t<t_0.
\ee
Now recall that (\ref{di-interim3}) is uniform in $\ep$, and $\ln D_n(t)$ is
continuous at $t=0$.
Therefore, taking the limit $\ep\to 0$ in (\ref{di-interim3}) and using the
Fisher-Hartwig
asymptotics (\ref{FH}) for
$\ln D_n(0)$ gives the expression (\ref{expansion Dn}) of Theorem \ref{theorem:
Toeplitz}.

This concludes the proof of both  Theorem \ref{theorem: Toeplitz} and  Theorem
\ref{theorem: Toeplitz2}.

\section*{Acknowledgements}
The authors are grateful to Alexander Abanov, Yan Fyodorov, and Jon
Keating for encouraging our interest in transition asymptotics for
determinants. Tom Claeys is a Postdoctoral Fellow of the Fund for
Scientific Research - Flanders (Belgium), and was also supported by
the ESF program MISGAM. Alexander Its was supported in part by NSF
grant \#DMS-0701768. Igor Krasovsky was supported in part by EPSRC
grant \#EP/E022928/1.


\begin{thebibliography}{99}
\bibitem{Andreev} F. V. Andreev, On special solutions of the fifth Painlev\'e
equation, {\em J. Math. Sci.} {\bf 99}  (2000),  no. 1, 802--807.
\bibitem{AndreevKitaev} F.V. Andreev and F. V. Kitaev, On connection formulas for
the asymptotics of some special solutions of the fifth Painlev\'e equation. {\em J.
Math. Sci.} {\bf 99} (2000), no. 1, 808--815.
\bibitem{AndreevKitaev2} F.V. Andreev and A.V. Kitaev, Connection formulae for
asymptotics of the fifth Painlev\'e transcendent on the real axis.  {\em
Nonlinearity} {\bf 13}  (2000),  no. 5, 1801--1840.
\bibitem{BDJ}
    J. Baik, P. Deift, and K. Johansson,
    On the distribution of the length of the longest increasing subsequence
    of random permutations,
    {\em J. Amer. Math. Soc.} {\bf 12} (1999), 1119--1178.

\bibitem{BE} Bateman, Erdelyi. Higher transcendental functions,
New York: McGraw-Hill, 1953--1955

\bibitem{Basor} E. Basor, Asymptotic formulas for Toeplitz determinants, {\em Trans.
Amer. Math. Soc.}
{\bf 239}  (1978), 33--65.

\bibitem{BTpainleve} E. L. Basor and C. A. Tracy, Asymptotics of a tau-function and
Toeplitz
determinants with singular generating functions,
{\em International J. of Mod. Phys. A} {\bf 7}, Suppl. 1A (1992), 83--107

\bibitem{BS}
A. B\"ottcher, B. Silbermann, Toeplitz operators and determinants
generated by symbols with one Fisher-Hartwig singularity.
{\em Math. Nachr.}  {\bf 127}  (1986), 95--123

\bibitem{BS2}
A. B\"ottcher, B. Silbermann, Toeplitz matrices and determinants with Fisher-Hartwig symbols.
{\em J. Funct. Anal.} {\bf 63} (1985), 178--214

\bibitem{BW}
A. B\"ottcher, H. Widom, Two elementary derivations of the pure Fisher-Hartwig determinant.
{\em Int. Eq. Op. Th.} {\bf 53} (2005), 593--596

\bibitem{DIK2} P. Deift, A. Its, and I. Krasovsky, Toeplitz and Hankel determinants
with Fisher-Hartwig singularities [arXiv:0905.0443]

\bibitem{DKMVZ2}
    P. Deift, T. Kriecherbauer, K.T-R McLaughlin, S. Venakides, and X. Zhou,
    Uniform asymptotics for polynomials orthogonal with respect to
    varying exponential weights and applications to universality
    questions in random matrix theory,
    {\em Comm. Pure Appl. Math.} {\bf 52} (1999), 1335--1425.
\bibitem{DKMVZ1}
    P. Deift, T. Kriecherbauer, K.T-R McLaughlin, S. Venakides,
    and X. Zhou,
    Strong asymptotics of orthogonal polynomials with respect to
    exponential weights,
    {\em Comm. Pure Appl. Math.} {\bf 52} (1999), 1491--1552.
\bibitem{DZ}
        P. Deift and X. Zhou,
        A steepest descent method for oscillatory Riemann-Hilbert problems.
            Asymptotics for the MKdV equation,
        {\em Ann. Math.} {\bf 137} (1993), no. 2, 295--368.
\bibitem{Ehr} T. Ehrhardt, A status report on the asymptotic behavior of Toeplitz
determinants with Fisher-Hartwig singularities,
{\em Operator Theory: Adv. Appl.} {\bf 124}, 217--241 (2001).
\bibitem{FIKN} A.S. Fokas, A.R. Its, A.A. Kapaev, and V.Yu.
        Novokshenov, ``Painlev\'e transcendents: the Riemann-Hilbert
        approach'', AMS Mathematical Surveys and Monographs \textbf{128}
        (2006).
\bibitem{FIK}
    A.S. Fokas, A.R. Its, and A.V. Kitaev,
    The isomonodromy approach to matrix models in 2D quantum gravity,
    {\em Comm. Math. Phys.} {\bf 147} (1992), 395--430.
\bibitem{FokasMuganZhou}
        A.S. Fokas, U. Mugan, and X. Zhou,
        On the solvability of Painlev\'e I, III and V,
        {\em Inverse Problems} {\bf 8} (1992), no. 5, 757--785.
\bibitem{FokasZhou}
        A.S. Fokas and X. Zhou,
        On the solvability of Painlev\'e II and IV,
        {\em Comm. Math. Phys.} {\bf 144} (1992), no. 3, 601--622.
\bibitem{Abanov}
F. Franchini and A.G. Abanov,
 Asymptotics of Toeplitz determinants and the emptiness formation probability for
the XY spin chain,
{\em J. Phys. A: Math. Gen.} {\bf 38} (2005), 5069--5095.

\bibitem{GI} B. L. Golinskii and I. A. Ibragimov,
A limit theorm of G. Szeg\H o. (Russian)
{\em Izv. Akad. Nauk SSSR Ser. Mat.} {\bf 35} (1971), 408--427.
%

\bibitem{I}
I. A. Ibragimov,  A theorem of Gabor Szeg\H o. (Russian)
{\em Mat. Zametki}  {\bf 3} (1968) 693--702.

\bibitem{IK1} A. Its and I. Krasovsky, Hankel determinant and orthogonal polynomials
for the Gaussian
weight with a jump, {\em Contemp. Math.} {\bf 458} (2008), 215--247.

\bibitem{ITW2} A. Its, C. Tracy, H. Widom, Random words, Toeplitz determinants and
integrable systems.
II, {\em Phys. D}  {\bf 152/153}  (2001), 199--224.

\bibitem{J} M. Jimbo, Monodromy problem and the boundary condition for some
Painlev\'e equations,
{\em Publ. RIMS, Kyoto Univ.} {\bf 18} (1982), 1137--1161.

\bibitem{jm2} M. Jimbo and T. Miwa, Studies on holonomic quantum fields XVII,
{\em Proc. Japan Acad.} {\bf 56} A (1980), 405 -- 410.

\bibitem{Jo} K. Johansson,
On Szeg\H oÕs asymptotic formula for Toeplitz determinants and generalizations,
{\em Bull. Sci. Math.} (2) {\bf 112} (1988), no. 3, 257--304.

\bibitem{KMM}
    S. Kamvissis, K.D.T-R McLaughlin, and P.D. Miller,
    ``\,Semiclassical soliton ensembles for the focusing nonlinear
    Schr\"odinger equation",
    {Ann. Math. Studies} {\bf 154}, Princeton Univ. Press,
    Princeton (2003).

\bibitem{McCoy}
B. M. McCoy, The connection between statistical mechanics and quantum field theory
[arxiv: hep-th/9403084]

\bibitem{MTW}
B. M. McCoy, C. A. Tracy and T. T. Wu, Painlev\'e functions of the
third kind, {\em J. Math. Phys.} {\bf 18} (1977), 1058--1092

\bibitem{MT} B.M. McCoy and S. Tang,
Connection formulae for Painlev\'e functions. Solitons and coherent structures
(Santa Barbara, Calif., 1985),
{\em Phys. D} {\bf 18} (1986), no. 1-3, 190--196.
\bibitem{MW} B. M. McCoy and T. T. Wu, The two-dimensional Ising model.
Harvard Univ. Press: Cambridge MA, 1973.
\bibitem{Shukla} P. Shukla,
Level spacing functions and the connection problem of a fifth Painlev\'e transcendent,
{\em J. Phys. A} {\bf 28} (1995), no. 11, 3177--3195.

\bibitem{T} C. A. Tracy, Asymptotics of a tau function arising in the
two-dimensional Ising model, {\em Commun. Math. Phys.} {\bf 142}
(1991), 297--311.


\bibitem{W} H. Widom. Toeplitz determinants with singular generating
functions.
{\em Amer. J. Math.} {\bf 95} (1973), 333--383

\bibitem{WMTB}
T. T. Wu, B. M. McCoy, C. A. Tracy and E. Barouch, Spin-spin
correlation functions for the two-dimensional Ising model: Exact
theory in the scaling region, {\em Phys. Rev.} {\bf B13} (1976),
316--374

\end{thebibliography}
\end{document}